\documentclass[twocolumn, amsmath,amssymb, notitlepage, prr,raggedbottom]{revtex4-2}
\usepackage{amssymb, amsthm, amsmath,float,xcolor, complexity,comment, mathtools}
\usepackage{graphicx}
\newtheorem{theorem}{Theorem}

\newtheorem{lemma}{Lemma}
\newtheorem{definition}[theorem]{Definition}

\newcommand{\tinyspace}{\mspace{1mu}}
\newcommand{\microspace}{\mspace{0.5mu}}

\newcommand{\norm}[1]{\left\lVert\tinyspace#1\tinyspace\right\rVert}

\def\<{\langle}
\def\>{\rangle}
\def \lket {\left|}
\def \rket {\right\rangle}
\def \lbra {\left\langle}
\def \rbra {\right|}
\newcommand{\ket}[1]{\lket\microspace #1 \microspace\rket}
\newcommand{\bra}[1]{\lbra\microspace #1 \microspace\rbra}

\DeclarePairedDelimiter{\ceil}{\lceil}{\rceil}

\usepackage{bm}
\let\vec\bm 

\usepackage{xr-hyper}
\usepackage[colorlinks=true, urlcolor=blue,citecolor=blue,anchorcolor=blue]{hyperref}

\newcommand{\sgn}{\mathrm{sgn}}

\begin{document}

\title{Minimum entanglement protocols for function estimation}
\author{Adam Ehrenberg}
\thanks{These two authors contributed equally.}
\affiliation{Joint Center for Quantum Information and Computer Science, NIST/University of Maryland College Park, Maryland 20742, USA}
\affiliation{Joint Quantum Institute, NIST/University of Maryland College Park, Maryland 20742, USA}
\author{Jacob Bringewatt}
\thanks{These two authors contributed equally.}
\affiliation{Joint Center for Quantum Information and Computer Science, NIST/University of Maryland College Park, Maryland 20742, USA}
\affiliation{Joint Quantum Institute, NIST/University of Maryland College Park, Maryland 20742, USA}
\author{Alexey V. Gorshkov}
\affiliation{Joint Center for Quantum Information and Computer Science, NIST/University of Maryland College Park, Maryland 20742, USA}
\affiliation{Joint Quantum Institute, NIST/University of Maryland College Park, Maryland 20742, USA}
\date{\today}

\begin{abstract}
    We derive a family of optimal protocols, in the sense of saturating the quantum Cram\'{e}r-Rao bound, for measuring a linear combination of $d$ field amplitudes with quantum sensor networks, a key subprotocol of general quantum sensor network applications. We demonstrate how to select different protocols from this family under various constraints. Focusing primarily on entanglement-based constraints, we prove the surprising result that highly entangled states are not necessary to achieve optimality in many cases. Specifically, we prove necessary and sufficient conditions for the existence of optimal protocols using at most $k$-partite entanglement. We prove that the protocols which satisfy these conditions use the minimum amount of entanglement possible, even when given access to arbitrary controls and ancilla. Our protocols require some amount of time-dependent control, and we show that a related class of time-independent protocols fail to achieve optimal scaling for generic functions.
\end{abstract}
\maketitle

\section{Introduction}
Entanglement is a hallmark of quantum theory and plays an essential role in many 
quantum technologies. 
Consider single-parameter metrology, where one seeks to determine an unknown phase 
$\theta$ that is independently and identically coupled to $d$ sensors via a linear Hamiltonian $\hat{H}$. Given a probe state $\hat{\rho}$, evolution under $\hat{H}$ encodes $\theta$ into $\hat{\rho}$ where it can then be measured.
If the sensors are classically correlated the ultimate attainable uncertainty is the so-called standard quantum limit $\Delta\theta\sim 1/\sqrt d$ \cite{giovannetti2004quantum}, which can be surpassed only if the states are prepared in an entangled state  \cite{giovannetti2006quantum, pezze2009entanglement}; if $O(d)$-partite entanglement is used, the Heisenberg limit $\Delta\theta\sim1/d$ can be achieved \cite{hyllus2012fisher, toth2012multipartite, augusiak2016asymptotic}. The necessity of entanglement for optimal measurement has also been explored in numerous other contexts \cite{toth2014quantum, braun2018quantum}; for instance, in sequential measurement schemes (where one may apply the encoding unitary multiple times) \cite{luis2002phase, higgins2007entanglement}, in the presence of decoherence \cite{huelga1997improvement, escher2011general, boixo2012entangled, demkowicz2014using}, when the coupling Hamiltonian is non-linear \cite{boixo2007generalized, boixo2008quantum, tilma2010entanglement}, or in reference to resource theories for metrology \cite{gour2008resource, marvian2014extending, marvian2016coherence, zhang2017detecting}.

In this paper, we consider the amount of entanglement required to saturate the quantum Cram\'{e}r-Rao bound, which lower bounds the variance of measuring an unknown quantity~\cite{helstrom1976quantum, braunstein1994statistical, braunstein1996generalized, holevo2011probabilistic}, in the prototypical multiparameter setting of a quantum sensor network, where $d$ independent, unknown parameters $\vec{\theta}$ (boldface denotes vectors) are each coupled to a unique quantum sensor.
Specifically, we revisit the problem of optimally measuring a single linear function $q(\vec{\theta})$ \cite{proctor2017networked, eldredge2018optimal, ge2018distributed, proctor2018multiparameter, altenburg2018multi, rubio2020quantum, gross2020one, triggiani2021heisenberg, oh2021distributed, malitesta2021distributed}, which is a crucial element of optimal protocols for more general quantum sensor network problems
(the case of measuring one or multiple analytic functions \cite{qian2019heisenberg, bringewatt2021protocols} and the case where the parameters $\vec\theta$ are not independent \cite{qian2020optimal} reduce asymptotically to the linear problem considered here). Therefore, we focus on measuring a single linear function of independent parameters for ease of presentation while emphasizing that our results generalize.

Given the similarity of measuring a single linear function
to the single-parameter case and the fact that such functions of local parameters are global properties of the system, one might expect (provided all the local parameters non-trivially appear in $q$) that $d$-partite entanglement is necessary. This intuition is reinforced by the fact that all existing optimal protocols for this problem do, in fact, make use of $d$-partite entanglement \cite{proctor2017networked, eldredge2018optimal, gross2020one}. 

We show that such intuition is faulty and only holds in the case where
$q$ is approximately an average of the unknown parameters. 
In particular, we derive a family of protocols that saturate necessary and sufficient algebraic conditions to achieve optimal performance in this setting, and
we prove necessary and sufficient conditions on $q$ for the existence of optimal protocols using at most $(k< d)$-partite entanglement. The more uniformly distributed $q$ is amongst the unknown parameters, the more entanglement is required. We also consider other resources of interest, such as the average entanglement used over the course of the protocol, as well as the number of entangling gates needed to perform these protocols, and discuss optimizing them within our scheme.

Finally, we address the impracticality of certain assumptions that have typically been made in the more theoretically-focused literature on function estimation protocols. Specifically, we show that so-called probabilistic protocols fail to achieve the Heisenberg limit except for a narrow class of functions.

\section{Problem Setup}
We first briefly review the problem of measuring a linear function of unknown parameters in a quantum sensor network  \cite{eldredge2018optimal, proctor2017networked, proctor2018multiparameter, altenburg2018multi, rubio2020quantum, gross2020one}. Consider a network of $d$ qubit quantum sensors coupled to $d$ independent, unknown parameters $\vec\theta\in\mathbb{R}^d$ via a Hamiltonian of the form
\begin{equation}\label{eq:H}
    \hat{H}(s)=\sum_{i=1}^d\frac{1}{2}\theta_i\hat{\sigma}_i^{z} +\hat{H}_c(s),
\end{equation}
where $\hat{\sigma}_i^{x,y,z}$ are the Pauli operators acting on qubit $i$ and $\hat{H}_{c}(s)$ for $s\in[0,t]$ is any choice of time-dependent, $\vec\theta$-independent control Hamiltonian, potentially including coupling to an arbitrary number of ancilla. That is, $\hat{H}_{c}(s)$ accounts for any possible parameter-independent contributions to the Hamiltonian, including those acting on any extended Hilbert space with a (finite) dimension larger than that of the network of $d$ qubit sensors directly coupled to the unknown parameters~\footnote{Thus, the Hilbert space under consideration is a $(d+n_a)$-qubit Hilbert space of dimension $2^{d+n_a}$, where $n_a$ is the number of ancilla.}.
We encode the parameters $\vec\theta$ into a quantum state $\hat{\rho}$ via the unitary evolution generated by a Hamiltonian of this form for a time $t$. Given some choices of initial probe state, control $\hat{H}_{c}(s)$, final measurement, and classical post-processing, we seek to construct an estimator for a linear combination $q(\vec\theta)=\vec\alpha\cdot\vec\theta$ of the unknown parameters, where $\vec\alpha\in\mathbb{R}^d$ is a set of known coefficients. Throughout this paper, we assume without loss of generality that $\norm{\vec\alpha}_\infty=|\alpha_1|$.
Ref.~\cite{eldredge2018optimal} established that the fundamental limit for the mean square error $\mathcal{M}$ of an estimator for $q$ is 
\begin{equation}\label{eq:bound}
    \mathcal{M}\geq \frac{\norm{\vec\alpha}_\infty^2}{t^2},
\end{equation}
where $t$ is the total evolution time.

Eq.~(\ref{eq:bound}) is derived via the single-parameter quantum Cram\'{e}r-Rao bound~\cite{helstrom1976quantum, braunstein1994statistical, braunstein1996generalized, boixo2007generalized, holevo2011probabilistic}. 
This is somewhat surprising: while we seek to measure only a single quantity $q(\vec\theta)$, $d$ parameters control the evolution under Eq.~(\ref{eq:H}), so we do not \emph{a priori} satisfy the conditions for the use of the single-parameter quantum Cram\'{e}r-Rao bound. However, we can justify its validity for our system: consider an infinite set of imaginary scenarios, each corresponding to a choice of artificially fixing $d-1$ degrees of freedom and leaving only $q(\vec\theta)$ free to vary. Under any such choice, our final quantum state depends on a single parameter $q$, and we can apply the single-parameter quantum Cram\'{e}r-Rao bound. While this requires giving ourselves information that we do not have, additional information can only reduce $\mathcal{M}$, and, therefore, any such choice provides a lower bound on $\mathcal{M}$ when we do not have such information. To obtain the tightest possible bound there must be some choice(s) of artificially fixing $d-1$ degrees of freedom that gives us no (useful) information about $q(\vec \theta)$. We will derive algebraic conditions that characterize such choices.

Thus, we may apply the single-parameter quantum Cram\'{e}r-Rao bound
\begin{equation}\label{eq:boixo}
    \mathcal{M}\geq \frac{1}{\mathcal{F}(q)}\geq \frac{1}{t^2\norm{\hat{g}_q}_s^2},
\end{equation}
where $\mathcal{F}$ is the quantum Fisher information, $\hat{g}_q=\partial\hat{H}/\partial q$ (the partial derivative fixes the other $d-1$ degrees of freedom), and the seminorm $\norm{\hat{g}_q}_s$ is the difference of the largest and smallest eigenvalues of $\hat{g}_q$ \cite{boixo2007generalized}. For our problem, the best choice of fixing extra degrees of freedom---in the sense of yielding the tightest bound via Eq.~(\ref{eq:boixo})---gives $\norm{\hat{g}_q}_s^2=1/\norm{\vec\alpha}_\infty^2$, yielding Eq.~(\ref{eq:bound})~\cite{eldredge2018optimal}. The proof of this fact is provided in Appendix~\ref{app:generalization} for completeness.

\section{Conditions for Saturable Bounds}
While the argument above justifies applying the single-parameter bound in our multiparameter scenario, it offers no roadmap for constructing optimal protocols. The quantum Fisher information matrix $\mathcal{F}(\vec\theta)$ provides an information-theoretic solution to this issue. When calculating $\mathcal{F}(\vec\theta)$ we restrict to pure probe states, as the convexity of the quantum Fisher information matrix implies mixed states fail to produce optimal protocols~\cite{fujiwara2001akio, liu2019quantum}. For pure probe states and unitary evolution for time $t$ under the Hamiltonian in Eq.~(\ref{eq:H}), it has matrix elements~\cite{liu2019quantum}
\begin{align}\label{eq:fim}
    \mathcal{F}(\vec\theta)_{ij}&=4\left[\frac{1}{2}\langle\{\hat{\mathcal{H}}_i(t),\hat{\mathcal{H}}_j(t)\}\rangle-\langle\hat{\mathcal{H}}_i(t)\rangle\langle\hat{\mathcal{H}}_j(t)\rangle\right],
\end{align}
where $\{\cdot,\cdot\}$ denotes the anti-commutator and
\begin{equation}
    \hat{\mathcal{H}}_{i}(t) = -\int_{0}^{t}ds \hat{U}^{\dag}(s) \hat{g}_i \hat{U}(s),
\end{equation}
with $\hat{g}_i=\partial \hat{H}/\partial\theta_{i}=\hat\sigma^z_j/2$ and $\hat{U}$ the time-ordered exponential of $\hat{H}$. The expectation values in Eq.~(\ref{eq:fim}) are taken with respect to the initial probe state.

Choosing $d-1$ degrees of freedom to fix in hopes of using the single-parameter bound then corresponds to a basis transformation $\vec\theta\rightarrow\vec q$, where we take $q_1=q$ to be our quantity of interest, and the other arbitrary $q_{j>1}$ are the extra degrees of freedom. This basis transformation has a corresponding Jacobian $J$ such that $\mathcal{F}(\vec q)=J^\top\mathcal{F}(\vec\theta)J$. To obtain the bound in Eq.~(\ref{eq:bound}) and have no information about $q(\vec\theta)$ from the extra degrees of freedom $q_{j>1}$,
$\mathcal{F}(\vec q)$ must have the following properties:
\begin{align}\label{eq:condonfim}
    \mathcal{F}(\vec q)_{11}&=\frac{t^2}{\alpha_1^2},\\
    \mathcal{F}(\vec q)_{1i}&=\mathcal{F}(\vec q)_{i1}=0\quad (\forall\, i\neq 1) \label{eq:condonfimoffdiagonal}
\end{align}
(recall $|\alpha_1|=\norm{\vec\alpha}_\infty$ without loss of generality). 
Via the inverse basis transformation $\vec q\rightarrow \vec \theta$, we find Eqs.~(\ref{eq:condonfim})-(\ref{eq:condonfimoffdiagonal}) are satisfied if and only if
\begin{equation}\label{eq:saturabilitycond}
   \mathcal{F}(\vec\theta)_{1j} = \mathcal{F}(\vec\theta)_{j1} = \frac{\alpha_j}{\alpha_1}t^2,
\end{equation}
where we assume here and for the rest of the main text that $|\alpha_1|>|\alpha_j|$ $\forall j>1$ for ease of presentation. Our main result (see Theorem~\ref{thm:entanglement}) is 
unchanged by this assumption, although its proof and that of several other results becomes more tedious.
The explicit derivation of Eq.~(\ref{eq:saturabilitycond}), along with the generalization of our results beyond this assumption, is provided in Appendix~\ref{app:generalization}. 

Finally, we remark that the problem of function estimation is mathematically equivalent to the concept of nuisance parameters in the literature on classical~(c.f. \cite{amari2012differential}) and quantum estimation theory~\cite{yang2019attaining, suzuki2020nuisance, suzuki2020quantum}. One finds similarly derived bounds in these contexts~\footnote{For instance, the conditions in Eqs.~(\ref{eq:condonfim})-(\ref{eq:condonfimoffdiagonal}) are equivalent to the so-called global parameter orthogonality condition discussed in Sect. 5.5 of Ref.~\cite{suzuki2020quantum}.}. However, the protocols we now describe, and especially their entanglement features, are new to this work.

\section{A Family of Optimal Protocols} 
We now derive a family of protocols that achieve Eq.~(\ref{eq:saturabilitycond}). A particular protocol consists of preparing a pure initial state $\hat\rho_0=\ket{\psi(0)}\bra{\psi(0)}$, evolving $\hat\rho_0$ under the unitary generated by $\hat H(s)$ for time $t$, performing some positive operator-valued measurement,
and computing an estimator for $q$ from the measurement outcomes. Given $\hat\rho_0$ and $\hat H(s)$, $\mathcal{F}(\vec\theta)$ can be computed via Eq.~(\ref{eq:fim}).

The protocols we propose will use $\hat H_c(s)$ to coherently switch between probe states with different sensitivities to the unknown parameters $\vec\theta$, thereby accumulating an overall sensitivity to the unknown function of interest $q$. In particular, we consider the following set $\mathcal{T}$ of $N=3^{d-1}$ one-parameter families of cat-like states:
\begin{equation}\label{eqn:state_definition_1}
    \ket{\psi(\vec\tau; \varphi)}=\frac{1}{\sqrt{2}}\left(\ket{\vec\tau}+e^{i\varphi}\ket{-\vec\tau}\right),
\end{equation}
where each family of states is labeled by a vector $\vec\tau\in\{0,\pm 1\}^d$ such that
\begin{equation}\label{eqn:state_definition_2}
    \ket{\vec\tau}=\bigotimes_{j=1}^d\begin{cases}
    \ket{0},& \tau_j\neq -1\\
    \ket{1},& \tau_j=-1\\
    \end{cases},
\end{equation}
and $\varphi\in\mathbb{R}$ parameterizes individual states in the family.
We require that $\tau_1=1$, as any optimal protocol must always be sensitive to this most important parameter; see Lemma~\ref{lem:f11} in Appendix~\ref{sec:a_useful_lemma}.
Each of the probe states described in Eqs.~(\ref{eqn:state_definition_1}) and (\ref{eqn:state_definition_2}) is a superposition of exactly two states in the $\hat{\sigma}^{z}$ basis (which we call ``branches''). Note that these states use no ancilla.

Our protocols proceed in three main stages: a state initialization stage, a parameter encoding stage, and, finally, a measurement stage. In the state initialization stage, we prepare the probe state $\ket{\psi(\vec{\tau};0)}$ that is then coupled to the parameters in the parameter encoding stage via a Hamiltonian of the form of Eq.~(\ref{eq:H}). During this parameter encoding stage, we use the control Hamiltonian to coherently switch between families of probe states at particular (optimized) times, such that the relative phase between the branches is preserved during the switches (that is, $\hat{H}_{c}(s)$ changes $\vec{\tau}$, but not $\varphi$). This can be done using finitely many CNOT and $\hat\sigma^x$ gates. We stay in the family of states $\ket{\psi(\vec\tau^{(n)};\varphi)}$ for time $p_n t$, where $p_n\in[0,1]$ such that $\sum_n p_n=1$. Here $n$ indexes some enumeration of the families of states in $\mathcal{T}$. There are three possibilities for the relative phase that qubit $j$ induces between the two branches due to the time spent in family $n$. If $\tau_{j}^{(n)} = 0$, then no relative phase is accrued because qubit $j$ is disentangled. If $\tau_{j}^{(n)} = 1$, the relative phase imprinted by $\hat{\sigma}_{j}^{z}/2$ is $p_{n}\theta_{j}t$, while if $\tau_{j}^{(n)} = -1$, the relative phase is $-p_{n}\theta_{j}t$. Thus, the $j$-th qubit always induces a relative phase of $p_{n}\tau_{j}^{(n)}\theta_{j}t$. Accounting for all qubits, being in family $n$ for time $p_{n} t$ induces a relative phase 
\begin{equation}
    \phi_{n} = \sum_{j}p_{n}t\tau_{j}^{(n)}\theta_{j}.
\end{equation}
Given some time-dependent probe $\ket{\psi(t)}$ which is in each family $\ket{\psi(\vec{\tau}^{(n)};\varphi)}$ for time $p_{n}t$, the total phase $\phi$ accumulated between the branches over the course of the entire parameter encoding stage of the protocol is
\begin{equation}\label{eq:totalphase}
    \phi = \sum_{n}\phi_{n} = \sum_{n} \sum_{j}p_{n}t\tau_{j}^{(n)}\theta_{j} = \sum_{j}(T\vec{p})_{j}\theta_{j}t,
\end{equation}
where we implicitly defined $\vec p=(p_1, \cdots, p_N)^\top$ and the $d \times N$ matrix $T$ with matrix elements $T_{mn}=\tau^{(n)}_m$. If $\vec{p}$ is chosen such that $T\vec p\propto \vec\alpha$ this total phase is $\propto qt$. More formally, choosing $\vec{p}$ such that
\begin{equation}\label{eq:prob}
    T\vec p = \frac{\vec \alpha}{\alpha_1}
\end{equation}
achieves the saturability condition in Eq.~(\ref{eq:saturabilitycond}).
Algebraic details of this calculation are provided in Appendix~\ref{sec:proof_of_13}. 

Any nonnegative solution (in the sense that $p_{n}\geq 0$ $\forall\,n$) to Eq.~(\ref{eq:prob}) 
specifies a valid set of states and evolution times satisfying Eq.~(\ref{eq:saturabilitycond}).
Because the system in Eq.~(\ref{eq:prob}) is highly underconstrained, such protocols do not necessarily use all $3^{d-1}$ families of states in $\mathcal{T}$. As an illustrative example, consider the solutions to Eq.~(\ref{eq:prob}) for two qubits. The available families of states are described by 
\begin{align}
    T = \begin{pmatrix}
    \vec{\tau}^{(1)} & \vec{\tau}^{(2)} & \vec{\tau}^{(3)}
    \end{pmatrix} =
    \begin{pmatrix}
    1 & 1 & 1 \\
    1 & -1 & 0
    \end{pmatrix}.
\end{align}
By Eq.~(\ref{eq:prob}), the fraction of time spent in each family of states must satisfy
\begin{align}
    p_1+p_2+p_3=1,\\
    p_1-p_2=\frac{\alpha_2}{\alpha_1}.
\end{align}

Solving in terms of $p_1$ leads to the 1-parameter family of solutions $p_2=p_1-\frac{\alpha_2}{\alpha_1}$ and $p_3=1+\frac{\alpha_2}{\alpha_1} -2p_1$, where $p_{n}\in[0,1]$ for all $n$.
Without loss of generality, assume $\alpha_1=1$. Then non-negativity is achieved by
\begin{equation}\label{eqn:example}
    p_1\in\begin{cases}
    \big[\alpha_2, \frac{1+\alpha_2}{2}\big] & \alpha_2\geq 0\\
    \big[0, \frac{1+\alpha_2}{2}\big] & \alpha_2< 0
    \end{cases}.
\end{equation}
There are many solutions satisfying these constraints. Of particular note, there is a two-family protocol that does not require using exclusively maximally entangled states: for $\alpha_{2} > 0$, let $p_1 = \alpha_{2}$ so that $p_2 = 0$ and $p_3 = 1-\alpha_{2}$; for $\alpha_{2} < 0$, let $p_1 = 0$ so that $p_2 = -\alpha_{2}$ and $p_3 = 1+\alpha_{2}$.

We refer to protocols achieving Eq.~(\ref{eq:prob}) (or, equivalently, Eq.~(\ref{eq:saturabilitycond})) as optimal. Note, however, that achieving these conditions is a property of the probe state(s) used and does not \textit{a priori} guarantee the existence of measurements to extract $q$. Therefore, we now move on to describing the third main stage of our protocols, which is the explicit measurement scheme: apply a sequence of $\hat{\sigma}_{i}^{x}$ and $\mathrm{CNOT}$ gates to the final state of a protocol to transform it into $1/\sqrt{2}(\ket{0}+e^{iqt/\alpha_{1}}\ket{1})(\ket{0\dots0})$. Then perform single qubit phase estimation to measure $q$~\footnote{It is worth pointing out that it is not strictly necessary to reduce the problem to single qubit phase estimation. The reason we consider disentangling all qubits is to reduce fully to the single qubit phase estimation problem of the robust phase estimation papers in Refs.~\cite{kimmel2015robust,kimmel2015robusterratum,belliardo2020achieving}, described below. However, one could apply essentially equivalent protocols by forgoing the disentangling of the qubits and simply performing parity measurements on the final cat-like state. Such parity measurements can be carried out by simply measuring all qubits individually.}. 

Such phase estimation is not as simple as it might appear, however. 
Because we are interested in how our error scales in the $t\rightarrow\infty$ limit, a naive approach loses track of which $2\pi$ interval the phase is in~\cite{higgins2009demonstrating, hayashi2018resolving, gorecki2020pi}. We could assume that this information is known \emph{a priori}~\cite{eldredge2018optimal}, but this is unjustified in practice as the required knowledge is of precision $\sim |\alpha_1|/t$, i.e. it is already within the Heisenberg limit. More realistically, starting with any $t$-independent prior knowledge of the unknown phase, we use the so-called robust phase estimation protocols from Refs.~\cite{kimmel2015robust, kimmel2015robusterratum, belliardo2020achieving} to saturate Eq.~(\ref{eq:bound}) up to a modest constant factor. {Such protocols work by optimally dividing the total time $t$ into $K$ stages with stage $k$ using a time $2\nu_k t_k$ such that $2\sum_{k=1}^K \nu_k t_k=t$. In each stage, one encodes the parameters into the state for a time $t_k$ and then makes a ($\hat{\sigma}^{x}$ or $\hat{\sigma}^{y}$) measurement. This is repeated $2\nu_k$ times in order to obtain an estimate of $q$, which in each stage becomes a more and more precise estimate. Provided the time of the final stage scales linearly with the total time, i.e., $t_K\sim t$, Heisenberg scaling in time is still achieved and we can estimate $q$ with a mean square error achieving the bound in Eq.~(\ref{eq:bound}) up to a constant factor. For completeness, we review this measurement scheme in more detail in  Appendix~\ref{sec:robust_phase_estimation}.

To summarize, a full optimal protocol is as follows:
\begin{enumerate}
    \item Using any relevant experimental desiderata and optimization algorithm, find a nonnegative solution $\vec{p}$ to Eq.~(\ref{eq:prob}).
    \item Restrict $\vec{p}$ to its $\overline{N}$ nonzero elements, and restrict $T$ to the corresponding columns. If desired, reorder the elements of $\vec{p}$ and the columns of $T$. The $\overline{N}$ $\vec{\tau}$ corresponding to the columns of $T$ will be the families of states used in the protocol. 
    \item Initialize a quantum state on the $d$ qubits to $\ket{0}^{\otimes d}$.
    \item Using CNOT and $\hat{\sigma}^{x}$ gates, prepare $ \ket{\psi(\vec\tau^{(1)}; 0)}$, the first state of the protocol. Couple the state to the Hamiltonian $\hat{H}$ and remain in this family for time $p_{1}t_k$, leading to state $\ket{\psi(\vec\tau^{(1)}; \phi_{1})}$, where $\phi_{1} = \sum_{j}p_{1}t_{k}\tau_{j}^{(1)}\theta_{j}$. Here, $t_k$ is the time required by the current step of the robust phase estimation protocol. 
    \item Using CNOT and $\hat{\sigma}^{x}$ gates, coherently switch to $\ket{\psi(\vec\tau^{(2)}; \phi_{1})}$ from $\ket{\psi(\vec\tau^{(1)}; \phi_{1})}$. Remain in this family for time $p_{2}t_k$, leading to state $\ket{\psi(\vec\tau^{(2)}; \phi_{1}+\phi_{2})}$, with  $\phi_{2} = \sum_{j}p_{2}t_{k}\tau_{j}^{(2)}\theta_{j}$. 
    \item Repeat this process for all states in the restricted $T$, staying in the family parameterized by $\vec{\tau}^{(n)}$ for time $p_{n}t_k$, leading to a final state $\ket{\psi(\vec\tau^{(\overline{N})}; qt_{k})}$. 
    \item Using CNOT and $\hat{\sigma}^{x}$ gates, convert this final state to $1/\sqrt{2}(\ket{0}+e^{iqt_k}\ket{1})\ket{0}^{\otimes d-1}$.
    \item Make a measurement on the first qubit of the final state (see Appendix~\ref{sec:robust_phase_estimation} for more details) and repeat starting from step 3. After $2 \nu_k$ repetitions, move to the next stage of the robust phase estimation protocol, and use an updated evolution time $t_k$. After a number of stages $K$ as prescribed by the robust phase estimation protocol, extract a final estimate of $q$ with a mean square error achieving the bound in Eq.~(\ref{eq:bound}) up to a constant factor.
\end{enumerate}

Having described the full details of the protocol, including the subtleties involved in subdividing the total time $t$ into different stages in order to implement robust phase estimation, in the rest of the paper, for simplicity of presentation, we will simply consider the total encoding time $t$ and act as if the parameters can be encoded into the state in one step, using evolution for this full time. This should be viewed as a notational shorthand such that $t$ can be replaced with the relevant $t_k$ at any given stage when implementing the full protocol.

\section{Minimum Entanglement Solutions}
We now focus on solutions from our family of protocols that require the minimum amount of entanglement.
Specifically, we prove necessary and sufficient conditions on $\vec\alpha$ for the existence of a protocol that uses at most $k$-partite entanglement. 
This is the primary technical result of our paper. We emphasize that, while the protocols in the previous section use a particular choice of controls that does not include ancilla qubits, Theorem~\ref{thm:entanglement} applies to any protocol making use of a Hamiltonian described via Eq.~(\ref{eq:H}).

\begin{theorem}[Main result]\label{thm:entanglement}
Let $q(\vec\theta)=\vec\alpha\cdot\vec\theta$. Without loss of generality, let $\norm{\vec\alpha}_\infty=|\alpha_1|$. Let $k\in\mathbb{Z}^+$ so that
  \begin{equation}\label{eq:condfull}
      k-1<\frac{\norm{\vec\alpha}_1}{\norm{\vec\alpha}_\infty} \leq k.
 \end{equation}
An optimal protocol to estimate $q(\vec\theta)$, where the parameters $\vec\theta$ are encoded into the probe state via unitary evolution under the Hamiltonian in Eq.~(\ref{eq:H}) requires at least, but no more than, $k$-partite entanglement. 
\end{theorem}
Theorem~\ref{thm:entanglement} justifies our claim that $d$-partite entanglement is not necessary unless $\norm{\vec{\alpha}}_{1}$ is large enough, i.e. in the case of measuring an average ($\alpha_i=\frac{1}{d}$ $\forall\, i$).
We now sketch the proof, providing full details in Appendix~\ref{sec:part2full}. The proof comes in two parts. First, using $k$-partite entangled states from the set of cat-like states considered above, we show 
the existence of an optimal protocol,
subject to the upper bound of Eq.~(\ref{eq:condfull}). 
Second, we show that, subject to the conditions in the theorem statement, there exists no optimal protocol using at most $(k-1)$-partite entanglement, proving the lower bound of Eq.~(\ref{eq:condfull}). 

\emph{Part 1.} Define $T^{(k)}$ to be the submatrix of $T$ with all columns $n$ such that $\sum_m |T_{mn}|>k$ are eliminated, which enforces that any protocol derived from $T^{(k)}$ uses only states that are at most $k$-partite entangled. Define System $A(k)$ as
 \begin{align}
     T^{(k)}\vec p^{(k)} &=\vec\alpha/\alpha_1, \label{eq:thm1a}\\
     \vec{p}^{(k)}&\geq0 \label{eq:thm1b}.
 \end{align}
Let $\vec\alpha'= \vec\alpha/\alpha_1$ and define System $B(k)$ as
\begin{align}
    (T^{(k)})^\top\vec y\geq 0 \label{eq:cond1},\\
    \langle\vec\alpha', \vec y\rangle < 0 \label{eq:cond2}.
\end{align}
By the Farkas-Minkowski lemma~\cite{farkas1902, dinh2014farkas}, System $A(k)$ has a solution if and only if System $B(k)$ does not, so
it is sufficient to show that System $B(k)$ 
does not have a solution if 
$\sum_{j>1}|\alpha_{j}'| \leq k-1$, where we used that $\alpha'_1=1$. This can be shown by contradiction.

\emph{Part 2.} The probe state must always be maximally sensitive to the first sensor qubit (see Lemma~\ref{lem:f11} in Appendix~\ref{sec:a_useful_lemma}), so $\mathcal{F}(\vec\theta)_{1j}$ only accumulates in magnitude when qubit $j$ is entangled with the first qubit (intuitively, Eq.~(\ref{eq:fim}) is similar to a connected correlator). Using this, we show that satisfying the condition in Eq.~(\ref{eq:saturabilitycond}) requires $\norm{\vec\alpha}_1/\norm{\vec\alpha}_\infty>k-1$. 
\hfill $\square $

Theorem~\ref{thm:entanglement} provides conditions for the existence of solutions to Eq.~(\ref{eq:prob}) with limited entanglement, but it is not constructive. To obtain an explicit protocol, simply solve the system of linear equations $T^{(k)}\vec p=\vec\alpha$. 

Of course, instantaneous entanglement is not the only resource that one might want to minimize. For instance, one might also be interested in minimizing average entanglement over the entire protocol. This possibility is considered in Section~\ref{sec:avg_entanglement}. Other, more general, resource restrictions can be handled by setting up a constrained optimization problem that picks out certain solutions to the system of linear equations $T^{(k)}\vec p=\vec\alpha$ subject to a cost function $\mathcal{E}(\vec p)$. For example, if certain pairs of sensors are easier to entangle than others, due to, for instance, their relative spatial location in the network, that could be encoded into $\mathcal{E}(\vec p)$. More complicated optimizations could also take into consideration the ordering of the states used in the protocols. For example, because our protocols require coherently applying CNOT gates to move between different families of entangled states, and these gates may be costly or error-prone resources, one might wish to find protocols that minimize the usage of these gates. We discuss this possibility and the potential tradeoff between minimizing entanglement and CNOT gates in Section~\ref{sec:CNOT}.

\section{Average Entanglement}\label{sec:avg_entanglement}
As mentioned above, one might also wish to minimize not just the size of the most-entangled family of states, but also the average entanglement used (given by weighting the size of each entangled family by the proportion of time that the family is used in the protocol). In this section (with some details deferred to Appendix~\ref{sec:non-echoed}), we show that there exists a class of optimal protocols, ones that we name ``non-echoed,'' that minimize this average entanglement. The formal definition is as follows:
\begin{definition}[Non-Echoed Protocols]
Consider some $\vec{\alpha}\in\mathbb{R}^d$ encoding a linear function of interest. Let $T$ be the matrix which describes our families of cat-like probe states, and let $\vec{p}$ specify a valid protocol such that $\vec p>0$ and $T\vec{p} = \vec{\alpha}/\norm{\vec{\alpha}}_{\infty}$. We say that the protocol defined by $\vec{p}$ is ``non-echoed'' if $\forall i$ such that $p_{i}$ is strictly greater than 0, $\sgn(T_{ij}) \in \{0, \sgn(\alpha_{j})\}$. 
\end{definition}
At any stage of a non-echoed protocol, letting the portion of the relative phase accumulated between the two branches of the probe state associated to the parameter $\theta_{i}$ be given by $c_{i}\theta_{i}$, two conditions must hold: (1) $|c_{i}| < |\alpha_{i}|$; (2) $\sgn(c_{i}) = \sgn(\alpha_{i})$. More intuitively, sensitivity to each parameter is accumulated ``in the correct direction'' at all times, meaning one does not use any sort of spin echo to produce a sensitivity to the function of interest, hence the name ``non-echoed.''

We now prove two useful statements about non-echoed protocols.
\begin{lemma}\label{lemma:avg_entanglement}
    Non-echoed protocols use minimium average entanglement.
\end{lemma}
\begin{proof}
    We start with $T\vec{p} = \vec{\alpha}/\norm{\vec{\alpha}}_{\infty}$. 
    Then
    \begin{align}
        \norm{\vec{\alpha}}_{1}/\norm{\vec{\alpha}}_{\infty} &= \sgn(\vec{\alpha})^{\top}(T\vec{p}) \nonumber \\
        &= (\sgn(\vec{\alpha})^{\top}T)\vec{p} = \vec{w}^{\top}\vec{p},
    \end{align}
    where we have defined $w_{j} = \sum_{i}|T_{ij}|$ to be the sum of the absolute value of the elements of the $j$th column of $T$. That is, $w_{j}$ represents how entangled the corresponding cat-like family of states is. But, then, clearly $\vec{w}^{\top}\vec{p}$ is the average entanglement of the entire protocol. Furthermore, the second half of the proof of Theorem~\ref{thm:entanglement}, given in Appendix~\ref{sec:part2full} shows that the minimum average entanglement of any optimal protocol is given by $\norm{\vec{\alpha}}_{1}/\norm{\vec{\alpha}}_{\infty}$ (see the discussion after the completion of the proof).
\end{proof}
The intuition behind this lemma is that if one always accumulates phase in the ``correct direction,'' then the total amount of entanglement used over the course of the protocol must be minimized, as any extra entanglement would lead to becoming overly sensitive to some parameter, which would require some sort of echo to correct. 

We further have the following theorem, which can be viewed as an extension of Theorem~\ref{thm:entanglement}.
\begin{theorem}\label{thm:avg_entanglement}
    For any $\vec{\alpha} \in \mathbb{R}^{d}$, there exists an optimal non-echoed protocol with minimum instantaneous entanglement for measuring $q = \vec{\alpha}\cdot\vec{\theta}$. 
\end{theorem}
The proof of this theorem is given in Appendix~\ref{sec:non-echoed}, and it proceeds in a very similar way to the proof of Theorem~\ref{thm:entanglement}. The main difference is that one also restricts the allowed state families to be those with the correct sign so as to be non-echoed. And, analogously to how one can find a protocol with minimum entanglement, one can also obtain a solution that minimizes average entanglement by restricting $T$ to only include columns such that $\sgn({T}_{ij})=\sgn(\alpha_i)$ for all $i,j$ and then solving the corresponding system of linear equations.

\section{CNOT Costs of Minimum Entanglement Protocols}\label{sec:CNOT}

We now address another resource of potential interest: how many entangling (CNOT) gates are required to perform our protocols with a focus on the minimum entanglement protocols. 

We will again assume, for simplicity, that $\norm{\vec{\alpha}}_{\infty} = \alpha_{1}  = 1 > |\alpha_{2}| \geq |\alpha_{3}| \geq \dots \geq |\alpha_{d}|$. Furthermore, without loss of generality, we will adopt the convention that an optimal protocol specified by a $\vec p\geq 0$ such that $T\vec{p}=\vec\alpha$ begins by preparing a state in the family described by the first column of $T$ and evolving for time $p_{1}t$, and then proceeds to the appropriate state (i.e., the one with phase $p_{1}t$) in the family described by the second column, then evolving for time $p_{2}t$, and so on, until eventually moving to the measurement state. If $p_{i} = 0$, the corresponding state family is skipped and not prepared. By construction, the number of CNOT gates needed to perform this protocol is the number of gates required to generate the first state, plus the number needed to convert from the first state to the second state, and so on. Finally, one should add the number of gates needed to prepare the measurement state, which disentangles all qubits, from the final probe state~\footnote{These gates are not strictly necessary. See footnote [48].}. The number of gates required to move from state $i$ to state $i+1$ corresponds to the number of elements of $\vec{\tau}_{i}$ that are $\pm 1$ but 0 in $\vec{\tau}_{i+1}$ and vice versa. In what follows, we will often consider only the gates that are used to convert between probe states (i.e., we will not consider the initial state preparation or final measurement preparation). This is physically motivated by the fact that these intermediate gates may be more difficult to perform or may be more susceptible to noise. 
Furthermore, assuming one is interested in the value of $q$ at some particular moment (and not, say, continuously), one might be free to prepare and purify the initial probe state in advance of the actual sensing task, which also justifies ignoring the initial CNOT cost. 

Assume that $\overline{N}$ states used in the protocol, i.e. $\vec{p}$ is such that it contains at most $\overline{N}$ nonzero elements. It is clear that at most $\mathcal{O}(\overline{N}^{2})$ CNOT gates are needed. However, this is not necessarily optimal. In fact, Ref.~\cite{eldredge2018optimal} provides a protocol that uses $d$ states and only $(d-1) = \mathcal{O}(d)$ intermediate CNOT gates. This ``disentangling protocol'' consists of using a maximally entangled Greenberger-Horne-Zeilinger state (up to $\hat{\sigma}^x$ rotations) for a time $|\alpha_{d}|t$, then disentangling the last qubit and using the $(d-1)$-entangled state for time $(|\alpha_{d-1}|-|\alpha_{d}|)t$ before disentangling the next-to-last qubit and so on until reaching the final state corresponding to $\vec{\tau} = (1, 0, \dots, 0)^\top$. This final state is used for time $(|\alpha_{1}|-|\alpha_{2}|)t = (1-|\alpha_{2}|)t$. The disentangling protocol does not minimize the instantaneous entanglement, but it does minimize average entanglement (as it is a non-echoed protocol---see Section~\ref{sec:avg_entanglement}). 

Even more interestingly, Ref.~\cite{eldredge2018optimal} also provides a protocol, which we refer to as the ``echoing'' protocol, that uses \emph{zero} intermediate CNOT gates. It proceeds by using $d$ exclusively maximally entangled states (thereby minimizing neither average nor, in most cases, instantaneous entanglement), but judiciously echoing away the extra sensitivity that this extra entanglement induces.

To illustrate these protocols in the language of the current paper, we provide $T$ and $\vec{p}$ (where, for simplicity of notation, we restrict $T$ and $\vec{p}$ to the states that are used for a non-zero fraction of time) for the case $d = 8$ and $\alpha_{i} > 0$:
\begin{widetext}
\begin{align}
    T^{\text{disentangling}} &= \begin{pmatrix}
    1 & 1 & 1 & 1 & 1 & 1 & 1 & 1 \\
    1 & 1 & 1 & 1 & 1 & 1 & 1 & 0 \\
    1 & 1 & 1 & 1 & 1 & 1 & 0 & 0 \\
    1 & 1 & 1 & 1 & 1 & 0 & 0 & 0 \\
    1 & 1 & 1 & 1 & 0 & 0 & 0 & 0 \\
    1 & 1 & 1 & 0 & 0 & 0 & 0 & 0 \\
    1 & 1 & 0 & 0 & 0 & 0 & 0 & 0 \\
    1 & 0 & 0 & 0 & 0 & 0 & 0 & 0
    \end{pmatrix},
    &\vec{p}^{\text{disentangling}} &= \begin{pmatrix}
    \alpha_{8}\\
    \alpha_{7}-\alpha_{8}\\
    \alpha_{6}-\alpha_{7}\\
    \alpha_{5}-\alpha_{6}\\
    \alpha_{4}-\alpha_{5}\\
    \alpha_{3}-\alpha_{4}\\
    \alpha_{2}-\alpha_{3}\\
    \alpha_{1}-\alpha_{2}
    \end{pmatrix} 
\end{align}
and
\begin{align}
    T^{\text{echoing}} &= \begin{pmatrix}
    1 & 1 & 1 & 1 & 1 & 1 & 1 & 1 \\
    1 & 1 & 1 & 1 & 1 & 1 & 1 & -1 \\
    1 & 1 & 1 & 1 & 1 & 1 & -1 & -1 \\
    1 & 1 & 1 & 1 & 1 & -1 & -1 & -1 \\
    1 & 1 & 1 & 1 & -1 & -1 & -1 & -1 \\
    1 & 1 & 1 & -1 & -1 & -1 & -1 & -1 \\
    1 & 1 & -1 & -1 & -1 & -1 & -1 & -1 \\
    1 & -1 & -1 & -1 & -1 & -1 & -1 & -1
    \end{pmatrix} ,
    &\vec{p}^{\text{echoing}} &= \begin{pmatrix}
    \frac{1+\alpha_{8}}{2}\\
    \frac{\alpha_{7}-\alpha_{8}}{2}\\
    \frac{\alpha_{6}-\alpha_{7}}{2}\\
    \frac{\alpha_{5}-\alpha_{6}}{2}\\
    \frac{\alpha_{4}-\alpha_{5}}{2}\\
    \frac{\alpha_{3}-\alpha_{4}}{2}\\
    \frac{\alpha_{2}-\alpha_{3}}{2}\\
    \frac{\alpha_{1}-\alpha_{2}}{2}
    \end{pmatrix}.
\end{align}
\end{widetext}
In the case of the disentangling protocol, the number of CNOTs needed is heavily dependent on the ordering of the states. For example, consider, instead, ordering the states in the following way:
\begin{equation}
    T^{\text{disentangling}} = \begin{pmatrix}
    1 & 1 & 1 & 1 & 1 & 1 & 1 & 1 \\
    1 & 0 & 1 & 1 & 1 & 1 & 1 & 1 \\
    1 & 0 & 1 & 0 & 1 & 1 & 1 & 1 \\
    1 & 0 & 1 & 0 & 1 & 0 & 1 & 1 \\
    1 & 0 & 1 & 0 & 1 & 0 & 1 & 0 \\
    1 & 0 & 1 & 0 & 1 & 0 & 0 & 0 \\
    1 & 0 & 1 & 0 & 0 & 0 & 0 & 0 \\
    1 & 0 & 0 & 0 & 0 & 0 & 0 & 0
    \end{pmatrix}.
\end{equation}
Here, the number of CNOTs required is now $(d-1) + (d-2) + \dots + 1 = \Theta(d^{2})$. Thus, it is not only the choice of states that affects the CNOT cost of a protocol, but also their ordering. Naively, finding an optimal set of states and their optimal ordering is a difficult problem, as if one finds a protocol using $\overline{N}$ states, there are $\overline{N}!$ orders to check.

While we were unable to find a general solution to this optimization problem, numerics allow us to provide a pragmatic analysis of the cost. To begin, we considered the naive approach of finding a random (non-echoed) minimum entanglement solution using $d$ states for random problem instances and, then, using this solution set, we brute-force searched over all column orderings of $T$ restricted to families of states specified by this solution to find an optimal ordering in terms of CNOT cost. This was done for $d\in[3,10]$ sensors with twenty random instances each. Without loss of generality, the random problem instances were taken to have all positive coefficients. We observe a CNOT cost scaling $\sim d^2$, indicating that a random minimum entanglement solution, even with optimal ordering, does not have the optimal linear in $d$ scaling. See Figure~\ref{fig:cnot_data}.

Consequently, more nuanced algorithms for finding a minimum entanglement solution with better CNOT costs are desirable. To this end, we considered a greedy algorithm that yields a $\Theta(d)$ CNOT cost whenever it does not fail. The algorithm works by building up the full sensitivity to one parameter before switching coherently to a new state family (in this way, it is non-echoed---see Section~\ref{sec:avg_entanglement}). Consequently, each time we switch to a new state, one sensor qubit can be disentangled and never re-entangled. In particular, we seek to build up sensitivity to the parameters according to their weight in $q$, i.e. we build up sensitivity to parameters going from the smallest corresponding $|\alpha_j|$ to the largest. The full algorithm is completed in at most $d$ steps~\footnote{Code is available upon request.}. 

However, this greedy algorithm can fail to produce a valid protocol, as it does not enforce the condition that $\norm{\vec p}_1=1$. This condition will be violated for some functions---typically those with many coefficients with approximately equal magnitude. Still, when it works, this algorithm succeeds in producing CNOT-efficient minimum entanglement protocols, as shown in Figure~\ref{fig:cnot_data}. Finding more general algorithms that always succeed for this task remains an interesting open problem.

\begin{figure}
    \centering
    \includegraphics[width=0.4\textwidth]{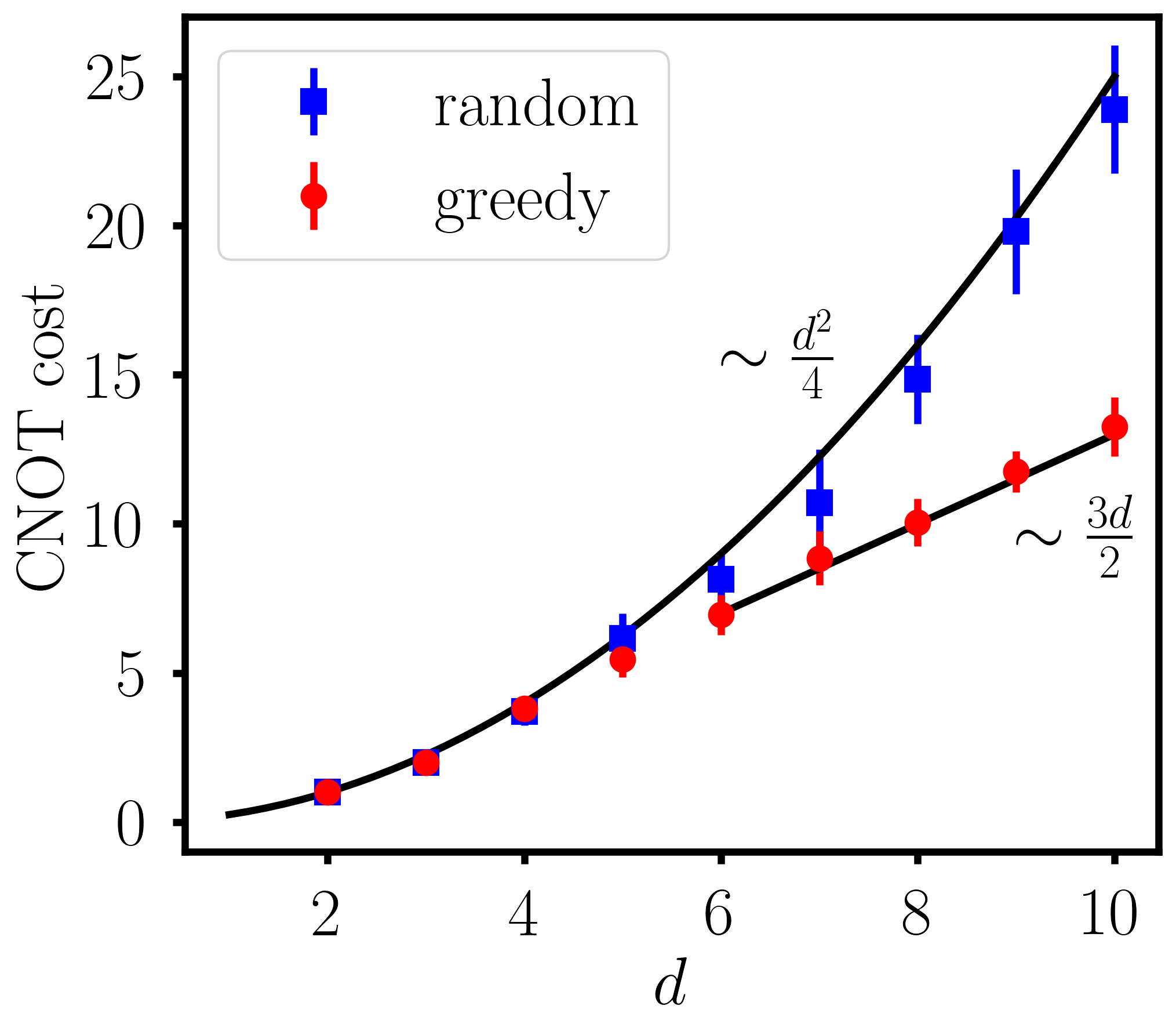}
    \caption{CNOT costs versus number of sensors $d$ for minimum entanglement protocols using $d$ optimally ordered states chosen either randomly or via the greedy algorithm described above. Twenty randomly chosen instances (that do not fail) to yield a valid protocol via the greedy algorithm. When it returns a valid protocol, the greedy algorithm recovers optimal linear scaling with $d$ for the CNOT cost, whereas randomly chosen states have quadratic scaling, even with optimal state ordering.}
    \label{fig:cnot_data}
\end{figure}

Independent of the algorithm used to minimize the CNOT count of an optimal protocol, the takeaway message is the same: there is an apparent tradeoff between entanglement- and gate-based resources. The disentangling protocol minimizes average entanglement, but not necessarily instantaneous entanglement, and requires only $\mathcal{O}(d)$ intermediate entangling gates; the echoing protocol uses maximal entanglement, but requires only single-particle intermediate gates. Protocols that minimize instantaneous entanglement do so at the cost of more intermediate entangling gates. Depending on the primary sources of error or the physical constraints on any given quantum sensor network implementation, one of these resources might be more important to minimize than the other. In general, determining the optimal CNOT scaling for protocols that minimize instantaneous and/or average entanglement is a crucial open question for future work.

\section{Time-Independent Protocols}
Another approach to constructing protocols is to use so-called probabilistic protocols. These protocols eschew control and instead exploit the convexity of the quantum Fisher information by staying in one family throughout any given run of the protocol, but by letting this family vary over different runs. Intuitively, each family is sensitive to a different function $q_n$ such that $q=\sum_{n=1}^{\overline{N}} p_n q_n$, where $\overline{N}$ is the number of families from $\mathcal{T}$ used in the protocol, and $p_{n}$ is the frequency that family $n$ is used. In this way, one can create an estimator for $q$ using those for $q_n$. In order to generate a Fisher information matrix satisfying Eq.~(\ref{eq:saturabilitycond})~\cite{eldredge2018optimal, gross2020one}, the $p_{n}$ should come from a solution to Eq.~(\ref{eq:prob}). These protocols have the advantage of requiring no control, but, unfortunately, suffer worse scaling with $d$ than ours for generic functions when the available resources are comparable. 

In particular, to fairly account for resources, we must fix a total time $t$ to perform \emph{all} stages of our protocol. Therefore, when considering a probabilistic protocol that uses multiple families from $\mathcal{T}$, but does not switch coherently between them, we must assign a time $t_n$ to family $n$ such that
\begin{equation}
    \sum_{j=1}^{\overline N} t_n=t.
\end{equation}
Note, we have used the fact that no stages of a probabilistic protocol with the families in $\mathcal{T}$ can be performed simultaneously. One could imagine protocols that parallelize the measurement of some $q_j$ that involve disjoint sets of sensors. However, such protocols are necessarily non-optimal given Lemma~\ref{lem:f11} in Appendix~\ref{sec:a_useful_lemma}, which says that any optimal protocol requires entanglement with the first qubit at all times.

We can bound the maximum of the Fisher information matrix element $\mathcal{F}(\vec\theta)_{11}$ obtainable via such a probabilistic protocol as
\begin{align}
    \max \mathcal{F}(\vec\theta)_{11} \leq  &\max_{p_n, t_n} \sum_{n=1}^{\overline N} p_n t_n^2, \nonumber \\
    &\text{subject to: } \sum_{n=1}^{\overline N} t_n=t, \nonumber\\
    & \qquad\qquad\quad \sum_{n=1}^{\overline N} p_n=1.
\end{align}
where we used that $\tau^{(n)}_1=1$ for all $n$. The inequality arises due to the fact that the maximization problem on the right hand side of the inequality does not enforce that $T\vec p =\vec\alpha/\alpha_1$. We could add this as an additional constraint, but it will not be necessary.

To perform the necessary optimization, consider the Lagrangian:
\begin{equation}
    \mathcal{L}=\sum_{n=1}^{\overline{N}} p_n t_n^2 + \gamma_1\left (t-\sum_{n=1}^{\overline N} t_n\right)+\gamma_2\left (1-\sum_{n=1}^{\overline N} p_n\right ),
\end{equation}
where $\gamma_1, \gamma_2$ are Lagrange multipliers. Therefore, we obtain the system of equations
\begin{align}
    2p_nt_n -\gamma_1 &= 0, \quad (\forall\, n), \nonumber\\ 
    t_n^2 -\gamma_2 &= 0, \quad (\forall\, n), \nonumber\\
    \sum_{n=1}^{\overline N} t_n&=t, \nonumber \\
    \sum_{n=1}^{\overline N} p_n&=1,
\end{align}
which can be solved to yield the solution
\begin{equation}
    \max_{p_n, t_n} \sum_{n=1}^{\overline N} p_n t_n^2 = \frac{t^2}{\overline N^2},
\end{equation}
for $p_n=1/\overline{N}$ and $t_n=t/\overline N$ for all $n$. Therefore,
\begin{equation}
    \mathcal{F}(\vec\theta)_{1j}\leq \frac{t^2}{\overline N^2}, \quad (\forall j),
\end{equation}
which clearly fails to achieve the saturability condition for $j=1$, unless $\overline N=1$, which is only possible for a very small set of functions (generic functions require $\overline{N}$ that scale nontrivially with $d$). Therefore, provided one considers cases where each $q_n$ must be learned sequentially (which is a requirement for any possibly optimal protocol via Lemma~\ref{lem:f11}), we fail to achieve saturability even up to a $d$-independent constant for generic functions via time-independent protocols.

Note that we have, for simplicity, again restricted ourselves to the case where $\vec{\alpha}$ has a single maximal magnitude element. The more general proof follows almost identically, with some notational overhead, when generalizing beyond this condition.

\section{Conclusion and Outlook}
We have proven that maximally entangled states are not necessary for the optimal measurement of a linear function with a quantum sensor network unless the function is sufficiently uniformly supported on the unknown parameters. While the uniformly distributed case has been considered extensively in the literature, as it provides the largest possible separation in performance between entangled and separable protocols, there is no \emph{a priori} reason why one should be interested in only these sorts of quantities. Our results demonstrate that while the precision gains to be had are less away from the uniformly distributed regime, the required resources are also less. This result is of particular relevance to the development of near-term quantum sensor networks, where creating large-scale entangled states may not be practical. Furthermore, while algebraic approaches like the one we consider here have been used before to generate bounds for the function estimation problem~\cite{eldredge2018optimal, qian2020optimal}, leveraging this approach to derive protocols that achieve these bounds subject to various experimental constraints is a new and widely applicable technique. We emphasize again that these results are also useful in more general settings, such as the measurement of analytic functions, as these measurements reduce to the case studied here \cite{qian2019heisenberg, qian2020optimal, bringewatt2021protocols}. 

To the best of the authors' knowledge, all information-theoretically optimal protocols for the estimation of a single linear function that are currently in the literature are subsumed by the framework that we develop in this work. What protocol one chooses to use will depend heavily on the experimental context; if decoherence is more problematic than the number of entangling gates that one must perform, then minimum entanglement protocols will be preferred to the conventional protocols. However, if decoherence is mild, but two-qubit gates introduce significant errors, then a protocol such as the echoing protocol presented in Ref.~\cite{eldredge2018optimal} will be preferred.
Consequently, the extent to which minimum entanglement protocols are more or less valuable than their more highly entangled counterparts depends on the details of the physical implementation of a quantum sensor network. Either way, the development of a framework to address these questions is, in of itself, an important contribution of this work.

We also briefly point out one more resource-related constraint of protocols that rely on time-dependent control (whether in the form of $\hat{\sigma}^{x}$ gates, CNOT gates, or others): these protocols require precise timing of the gate applications. Uncertainty in the timing leads directly to a systematic error in the function being measured. Importantly, however, this timing issue is a limitation of all known optimal protocols for the linear function estimation task (see e.g. Ref.~\cite{eldredge2018optimal}). We therefore view these limitations as more pertinent to experimental implementation than the theory of resource tradeoffs that we are considering here. 

So far, we have not discussed the situation where we are constrained to $k$-partite entanglement, but $k$ is not sufficient to achieve optimality (for any protocol) via Theorem~\ref{thm:entanglement}. We propose the following protocol for such a scenario: Let $R$ be a partition of the sensors into independent sets where we do not allow entanglement between sets and allow, at most, $k$-partite entanglement within each $r\in R$. Let $\vec\alpha^{(r)}$ denote $\vec\alpha$ restricted to $r$. Pick the optimal $R$ such that the condition of Theorem~\ref{thm:entanglement} is satisfied for all $r$; that is, we ensure that \emph{within} each independent set we obtain the optimal variance for the linear function restricted to that set. The result is a variance
\begin{equation}\label{eq:boundkrest}
     \mathcal{M}=\frac{1}{t^2}\sum_{r\in{R}} \norm{\vec\alpha^{(r)}}_\infty^2.
\end{equation}
The optimal $R$ is a partition of the sensors into contiguous sets (assuming for simplicity that $|\alpha_i|\geq |\alpha_j|$ for $i<j$) such that for all $r\in R$, $\sum_{i\in r} |\alpha_i|/\max_{i\in r}|\alpha_i|\leq k$, satisfying Theorem~\ref{thm:entanglement}.  
We conjecture that this protocol is optimal, and it is clearly so if partitioning the problem into independent sets is optimal. However, one could imagine protocols that use different partitions for some fraction of the time. Intuitively, this should not improve the performance, but we leave analyzing this as an open question. 

Finally, no optimal time-independent protocols for arbitrary linear functions exist in the literature. Finding such protocols (or proving their non-existence) remains an open problem of interest. 

\vspace{1em}
\acknowledgments

We thank Michael Foss-Feig, Zachary Eldredge, Tarushii Goel, Pradeep Niroula, Luis Pedro Garc\'ia-Pintos, and Michael Gullans for helpful discussions. We thank the anonymous referee who pointed out that the CNOT cost of our minimum entanglement protocols deserved a more detailed analysis. This work supported in part by AFOSR MURI, AFOSR, DARPA SAVaNT ADVENT, NSF PFCQC program, ARO MURI, DoE ASCR Accelerated Research in Quantum Computing program (award No.~DE-SC0020312), NSF QLCI (award No.~OMA-2120757), U.S.~Department of Energy Award No.~DE-SC0019449, and the DoE ASCR Quantum Testbed Pathfinder program (award No.~DE-SC00119040). Support is also acknowledged from the U.S.~Department of Energy, Office of Science, National Quantum Information Science Research Centers, Quantum Systems Accelerator.  This research was supported in part by the Heising-Simons Foundation, the Simons Foundation, and National Science Foundation Grant No.~NSF PHY-1748958. J.B.~acknowledges support by the U.S.~Department of Energy, Office of Science, Office of Advanced Scientific Computing Research, Department of Energy Computational Science Graduate Fellowship (award No.~DE-SC0019323).

\bibliography{main}

\onecolumngrid

\begin{appendix}
\section{A Useful Lemma Regarding Optimal Probe States}\label{sec:a_useful_lemma}
In this Appendix, we prove a useful lemma restricting the structure of the probe state for an optimal protocol. 

\begin{lemma}\label{lem:f11}
Any optimal protocol, independent of the choice of control, requires that $\langle \hat{\mathcal{H}}_{1}(t)\rangle = 0$, where $\mathcal{H}_{1}(t)$ is the time-evolved generator of the first parameter and the expectation value is taken with respect to the initial probe state. Further the probe state must be of the form
\begin{equation}
    \ket{\psi} = \frac{\ket{0}\ket{\varphi_{0}} + e^{i\phi}\ket{1}\ket{\varphi_{1}}}{\sqrt{2}},
\end{equation}
for all times $s \in [0, t]$, where $\phi, \ket{\varphi_{0}}, \ket{\varphi_{1}}$ are arbitrary states on the $d-1$ remaining sensor qubits plus, potentially, the arbitrary number of ancilla---they can be $s$-dependent.
\end{lemma}
\begin{proof}
Consider the expression for the matrix elements of the quantum Fisher information matrix at time $t$ (Eq.~(4) of the main text):
\begin{equation}\label{eq:Fij}
     \mathcal{F}(\vec\theta)_{ij}=4[\frac{1}{2}\langle\{\hat{\mathcal{H}}_i(t),\hat{\mathcal{H}}_j(t)\}\rangle-\langle\hat{\mathcal{H}}_i(t)\rangle\langle\hat{\mathcal{H}}_j(t)\rangle],
\end{equation}
where the expectation values are taken with respect to the initial probe state $\ket{\psi(0)}$. Using the integral form of $\hat{\mathcal{H}}_j(t)$ (Eq.~(5) of the main text), we can write
\begin{align}
    \mathcal{F}(\vec\theta)_{11}&=4\mathrm{Var}\left[\hat{\mathcal{H}}_1(t)\right] \\
    &= 4\left[\int_0^t ds \int_0^t ds' \bra{\psi(0)}\hat{U}^\dagger(s)\hat g_1 \hat{U}(s)\hat{U}^\dagger(s')\hat g_1 \hat{U}(s')\ket{\psi(0)}\right]-4\left[\int_0^t ds \bra{\psi(0)}\hat{U}^\dagger(s)\hat g_1 \hat{U}(s)\ket{\psi(0)}\right]^2 \\
    &= 4 \int_0^t ds \int_0^t ds' \mathrm{Cov}_{\ket{\psi(0)}}[\hat{g}_1(s),\hat{g}_1(s')],
\end{align}
where we recall
\begin{equation}
    \hat{g}_1(s):= \hat{U}^\dagger(s) \hat{g}_1 \hat{U}(s),
\end{equation}
and $\hat{g}_{1} = \partial \hat{H}/\partial\theta_{1}$ is the initial generator with respect to the first parameter. Once again, the covariance is with respect to the initial probe state $\ket{\psi(0)}$. We can then upper bound this as
\begin{align}
    \mathcal{F}(\vec\theta)_{11}(t)&\leq 4\int_0^t ds \int_0^t ds'  \sqrt{\mathrm{Var}_{\ket{\psi(0)}}[\hat{g}_1(s)]\mathrm{Var}_{\ket{\psi(0)}}[\hat{g}_1(s')]} \label{eq:begin_ineq} \\
    &= 4\left[\int_0^t ds \sqrt{\mathrm{Var}_{\ket{\psi(0)}}[\hat{g}_1(s)]}\right]^2  \\
    &\leq \left[\int_0^t ds \norm{\hat g_1}_s\right]^2 \label{eq:end_ineq}\\
    &= t^2 \norm{\hat g_1}_s^2 \\
    &= t^2,
\end{align}
where the first inequality bounds the covariance as the square root of the product of the variances, the second inequality bounds the standard deviation of an operator by half the seminorm \cite{boixo2007generalized}, and the final equality uses the fact that $\hat{g}_{1} = \hat{\sigma}_{1}^{z}/2$ has seminorm 1~\footnote{Note that the above block of equations relies on the fact that we are using the fixed Hilbert space of qubit sensors. Were one to extend this derivation to photonic sensors with indefinite particle number, the results would not immediately follow.}.

Via Eq.~(\ref{eq:saturabilitycond}) of the main text (rigorously derived in Appendix~\ref{app:generalization}) we know that an optimal protocol must have $\mathcal{F}_{11}(\vec{\theta})(t) = t^{2}$. Therefore, an optimal protocol must saturate the inequalities in Eq.~(\ref{eq:begin_ineq}) and Eq.~(\ref{eq:end_ineq}). Eq.~(\ref{eq:end_ineq}) is saturated when  $\mathrm{Var}[\hat{g}_1(s)]=\norm{\hat g_1(s)}_s=\norm{\hat g_1}_s$ for all $s$. This holds if and only if $\ket{\psi(0)}=\frac{1}{\sqrt{2}}\left(\ket{\lambda_\mathrm{min}}+e^{i\phi}\ket{\lambda_\mathrm{max}}\right)$, where $\ket{\lambda_\mathrm{min}}$ and $\ket{\lambda_\mathrm{max}}$ are the eigenstates corresponding to the minimum and maximum eigenvalues of $\hat g_1(s)$ for all $s\in[0,t]$ and $\phi$ is an arbitrary phase. Given this condition, $\hat g_1(s)$ and $\hat g_1(s')$ act identically on the state $\ket{\psi(0)}$ and consequently are fully correlated when one considers the covariance of these operators with respect to the state. The Cauchy-Schwarz inequality in Eq.~(\ref{eq:begin_ineq}) is immediately saturated as well.

Importantly, under this condition on the probe state, any operator in the one-parameter family $\hat g_1(s)=\hat{U}^\dagger(s)\hat g_1 \hat{U}(s)$ acts identically on $\ket{\psi(0)}$ (the unitary does not change the eigenvalues, and the eigenstates are shared by all $\hat{g}_{1}(s)$, as argued above). Thus, one can freely substitute any operator in the one-parameter family $\hat g_1(s)=\hat{U}^\dagger(s)\hat g_1 \hat{U}(s)$ for another. Therefore, for such an optimal probe state,
\begin{equation}
    \langle\mathcal{H}_1(t)\rangle=-\int_0^t ds \bra{\psi(0)}\hat g_1(s)\ket{\psi(0)}=t\langle \hat g_1\rangle = 0
\end{equation}
because $\hat{g}_{1} \propto \hat{\sigma}_{1}^{z}$ and, consequently, by the argument that we can replace $\hat{g}_{1}$ by $\hat{g}_{1}(s)$ when acting on the probe state,
\begin{equation}\label{eq:<g_1>=0}
     \bra{\psi(s)}\hat g_1\ket{\psi(s)}=0\quad (\forall s).
\end{equation}
The statement of the lemma immediately follows.
\end{proof}
Note that Lemma~\ref{lem:f11} holds for any optimal protocol, not just those using our cat-like states. However, it also justifies our choice of probe states and why we specifically set $\tau_{1} = 1$ for all $\vec\tau$ (i.e., to maintain an equal superposition between $\ket{0}$ and $\ket{1}$ on the first qubit).

\section{Proof of the Optimality of Cat-State Protocols }\label{sec:proof_of_13}
In this Appendix, we will rigorously prove the optimality of the time-dependent protocols considered in the main text. In particular, we show that the Fisher information matrix condition for saturability in Eq.~(\ref{eq:saturabilitycond}) of the main text is satisfied by solutions to Eq.~(\ref{eq:prob}) of the main text when we consider protocols that use $\hat\sigma^x$ and CNOT controls to switch between families of cat-like states in $\mathcal{T}$. That is, we show the following mapping between saturability conditions:
\begin{equation}\label{eq:mapconditions}
     T\vec p = \frac{\vec\alpha}{\alpha_1} \quad \implies \quad \mathcal{F}(\vec\theta)_{1j}=\frac{\vec\alpha}{\alpha_1}t^2, 
\end{equation}
where we recall that we have assumed that $|\alpha_1|=\norm{\vec\alpha}_\infty > |\alpha_j|$ for all $j$ (in Appendix~\ref{app:generalization}, we will generalize beyond the assumption of a single maximum magnitude $\alpha_j$ at the cost of some notational inconvenience). 

Using Lemma \ref{lem:f11}, we can show that for \textit{any} optimal protocol (i.e., not just those using our cat-like states)
\begin{align}
    \mathcal{F}(\vec\theta)_{1j}&=2\langle \{\hat{\mathcal{H}}_1,\hat{\mathcal{H}}_j\}\rangle\\
    &=2\int_0^t ds \int_0^t ds' \bra{\psi(0)}\{\hat g_1(s),\hat{U}^\dagger(s')\hat g_j\hat{U}(s')\}\ket{\psi(0)}\\
    &=2 \int_0^t ds \int_0^t ds' \bra{\psi(0)}\{\hat g_1,\hat{U}^\dagger(s')\hat g_j\hat{U}(s')\}\ket{\psi(0)}\\
    &=2t \int_0^t ds' \bra{\psi(0)}\{\hat g_1,\hat{U}^\dagger(s')\hat g_j\hat{U}(s')\}\ket{\psi(0)}\\
    &=2t \int_0^t ds' \bra{\psi(0)}\{\hat g_1(s'),\hat{U}^\dagger(s')\hat g_j\hat{U}(s')\}\ket{\psi(0)}\\
    &=4t \int_0^t ds' \bra{\psi(s')}\hat g_1\hat g_j\ket{\psi(s')} \label{eq:psis}\\
    &=t \int_0^t ds' \bra{\psi(s')}\hat \sigma^z_1\hat \sigma^z_j\ket{\psi(s')}.\label{eq:f1j_final}
\end{align}
The third and fifth equalities come from the argument in the proof of Lemma~\ref{lem:f11} that we may replace $\hat{g}_{1}(s)$ with $\hat{g}_{1}$ (and vice versa) when acting on optimal probe states. The penultimate equality is just a consequence of the commutativity of the initial generators. 

We now apply these general results to our specific protocols. Saturating the initial Fisher information conditions in Eq.~(\ref{eq:mapconditions}) implies that we must show
\begin{equation}\label{eq:s22}
    \int_0^t ds' \bra{\psi(s')}\hat \sigma^z_1\hat \sigma^z_j\ket{\psi(s')} = \frac{\alpha_{j}}{\alpha_{1}}t.
\end{equation}
Let the gates in our protocols be labeled as $\hat G_i$ where $\hat G_i$ is either a CNOT or $\hat\sigma^x$ gate. The gate $\hat G_i$ is applied at a time $s=t_i^*$. Then, for $s \in (t_{k}^{*}, t_{k+1}^{*})$, we can write the time-dependent state as
\begin{equation}
    \ket{\psi(s)} = \ket{\psi(\vec{\tau}^{(k)}; \varphi)} \equiv \prod_{i=0}^{k}\hat{G}_{i}\ket{\psi(\vec{\tau}^{(0)};\varphi)},
\end{equation}
where $\ket{\psi(\vec{\tau}^{(0)};0)}$ is the initial state of the protocol, $\varphi$ is the relative phase between the two branches of the state that has accumulated up to time $s$, and, therefore, $\ket{\psi(\vec{\tau}^{(k)};\varphi)}$ is the state produced after applying the first $k$ gates. Because our protocols explicitly use only $\hat{\sigma}^{x}$ and $\mathrm{CNOT}$ gates to move between families in $\mathcal{T}$, we have that $\ket{\psi(\vec{\tau}^{(k)};\varphi)}=(\ket{0}|\chi_{0}^{(k)}\rangle + e^{i\varphi} \ket{1}|\chi_{1}^{(k)}\rangle)/\sqrt{2}$, and
\begin{equation}
       \int_0^t ds' \bra{\psi(s')}\hat \sigma^z_1\hat \sigma^z_j\ket{\psi(s')} = \sum_{i=0}^{n}(t_{i+1}^{*}-t_{i}^{*})\tau_{j}^{(i)},
\end{equation}
where we implicitly define $t_{0}^{*} = 0$ and $t_{n+1}^{*}=t$ as the initial and final times of the protocol and $|\chi_{0}^{(k)}\rangle$ and $|\chi_{1}^{(k)}\rangle$ are some states defined on the Hilbert space which excludes the first qubit sensor.  The time $t_{i+1}^{*}-t_{i}^{*}$ corresponds to the time we are in the probe family $\ket{\psi(\vec{\tau}^{(i)};\varphi)}$, which in our protocols is $p_{i}t$. Thus, to satisfy the Fisher information conditions, we need
\begin{equation}
    \sum_{i} p_{i}\tau_{j}^{(i)} = \frac{\alpha_{j}}{\alpha_{1}} \implies (T\vec{p})_{j} = \frac{\alpha_{j}}{\alpha_{1}}.
\end{equation}
This formally proves optimality of our time-dependent protocols that satisfy $T\vec{p} = \vec{\alpha}/\alpha_{1}$.

\section{Review of Robust Phase Estimation}\label{sec:robust_phase_estimation}
In this Appendix, we review, for completeness, the phase estimation protocols of Refs.~\cite{kimmel2015robust, kimmel2015robusterratum, belliardo2020achieving} described in the main text as a method to extract the quantity of interest $q$ from the state
\begin{equation}\label{eq:finstate}
    1/\sqrt{2}(\ket{0}+e^{iqt/\alpha_{1}}\ket{1})(\ket{0\dots0}),
\end{equation}
which is the final state obtained from our family of optimal protocols. 

Again, when we refer to our protocols as optimal, we mean this in the sense that our protocols achieve the conditions on the quantum Fisher information matrix that allow the maximum possible quantum Fisher information with respect to the parameter $q$ to be obtained. However, to completely specify the procedure by which one obtains the quantity $q$, an explicit phase estimation protocol is needed. As explained in the main text, such a task is complicated by the fact that for large times and/or small $\alpha_1=\norm{\vec\alpha}_\infty$, it is unclear what $2\pi$ interval the relative phase between the branches of Eq.~(\ref{eq:finstate}) is in~\cite{higgins2009demonstrating, hayashi2018resolving}. The phase estimation protocols of Refs.~\cite{kimmel2015robust, kimmel2015robusterratum, belliardo2020achieving} demonstrate how to optimize resources to deal with this issue, while still saturating the single-shot bound in Eq.~(2) of the main text up to a small $d$- and $t$-independent constant. In particular, such protocols allow us to reach a mean square error of
\begin{equation}\label{eq:constoverhead}
    \mathcal{M} = \frac{c^2\norm{\vec\alpha}_\infty^2}{t^2},
\end{equation}
for some small (explicitly known) constant $c$. Ref.~\cite{gorecki2020pi} proves that this constant factor $c^{2}$ in Eq.~(2) can be reduced to, at best, $\pi^2$.

While reviewing such phase estimation protocols, we follow the presentation of Ref.~\cite{belliardo2020achieving}, which corrects a few minor errors in Ref.~\cite{kimmel2015robust}, as noted in the corresponding erratum~\cite{kimmel2015robusterratum}. We refer the reader to Ref.~\cite{belliardo2020achieving} for further details. Conveniently, by putting the final state into the form of Eq.~(\ref{eq:finstate}), we have reduced this problem completely to the single qubit, multipass version of the  problem described in that reference. Consequently, everything follows practically identically to their presentation.

Consider dividing the total time $t$, which is the relevant resource in our problem, into $K$ stages where we evolve for a time $M_j\delta t$ in the $j$-th stage ($\delta t$ is some small basic unit of time and $M_j\in\mathbb{N}$). We assume that we have ($d,t$)-independent, prior knowledge of $q$ such that we can set $\delta t$ to satisfy
\begin{equation}
   \frac{\delta t q}{\norm{\vec\alpha}_\infty}\in[0,2\pi).
\end{equation}
In the $j$-th stage, using one of our protocols for a time $M_j\delta t$, we prepare $2\nu_j$ independent copies of the state
\begin{equation}
    \ket{\psi_j}=\frac{1}{\sqrt{2}}\left(\ket{0}+e^{iqM_j\delta t/\norm{\vec\alpha}_\infty}\ket{1}\right)\ket{0...0},
\end{equation}
 From now on we will drop the $d-1$ qubit sensors in the state $\ket{0...0}$, as they are irrelevant; however, it is worth noting that it is not necessary to put the state in this form before performing measurements. We do so to make the comparison to Ref.~\cite{belliardo2020achieving} particularly transparent. We then perform a single-qubit measurement on the first qubit sensor of each of these state copies, yielding $2\nu_j$ measurement outcomes, which we can use to estimate $q$. The total time of this $K$ stage protocol is consequently given by
\begin{equation}
    t=2\sum_{j=1}^K\nu_jM_j \delta t.
\end{equation}

Given this setup, we choose single-qubit measurements and optimize the choice of $\nu_j, M_j$ per stage so that we can learn $q$ bit by bit, stage by stage, in such a way that optimal scaling in $d$, $t$ is still obtained [Eq.~(\ref{eq:constoverhead})]. In particular, consider making two measurements, each $\nu_j$ times per stage (thus explaining the factor of two we introduced earlier): (i) a $\hat\sigma^{x}$ measurement and (ii) a $\hat\sigma^{y}$ measurement. These measurements each give us outcomes that are Bernoulli variables (i.e. with values $\in\{0,1\}$) with outcome probabilities
\begin{align}
    p^{(x)}(0)&=\frac{1+\cos \left(M_j q \delta t/\norm{\vec\alpha}_\infty\right)}{2}, \nonumber \\
    p^{(x)}(1)&=1-p^{(x)}(0), \nonumber \\
    p^{(y)}(0)&=\frac{1+\sin \left(M_j q \delta t/\norm{\vec\alpha}_\infty\right)}{2}, \nonumber \\
    p^{(y)}(1)&=1-p^{(y)}(0),
\end{align}
where the first two probabilities are for the $\hat\sigma^{x}$ measurement and the latter two are for the $\hat\sigma^{y}$ measurement. Using both of these measurements allows us to resolve the two-fold degeneracy in the phase $q M_j \delta t/\norm{\vec\alpha}_\infty$ within a given $[0,2\pi)$ interval that would arise from, e.g., a $\hat\sigma^{x}$ measurement alone. The observed probabilities of obtaining $0$ for the $\hat\sigma^{x}$ and $\hat\sigma^{y}$ are independent random variables that converge in probability to their associated expectation values for $\nu_j\rightarrow\infty$. These measurements are non-adapative, which makes this particular phase estimation protocol especially appealing. 

At each stage, we extract an estimator $\tilde\phi$ of $\phi:=M_j q\delta t/\norm{\vec\alpha}_\infty$ as
\begin{equation}
    \tilde\phi:=\mathrm{atan2}(2f_0^{(y)}-1, 2f_0^{(x)}-1)\in[0,2\pi),
\end{equation}
where $\mathrm{atan2}$ is the 2-argument arctangent with range $[0, 2\pi)$. In the limit $\nu_j\rightarrow\infty$, this estimator indeed converges to $\phi$, but the ``magic'' of this phase estimation scheme lies in the correct reprocessing of data stage-by-stage so that $\nu_j$ can be kept $(d,t)$-independent. Ref.~\cite{belliardo2020achieving} demonstrates rigorously that picking $M_j=2^{j-1}$ for $j\in\{1,\cdots, K\}$ and optimizing over $\nu_j$ one can, at each stage, estimate $q/\norm{\vec\alpha}_\infty$ with a confidence interval of size $2\pi/(3\times 2^{j-1})$ so that in each stage we learn another bit of this quantity. The results of this optimization are $\nu_j$ that decrease linearly with the step $j$ so that as the time spent in a stage grows, the statistics we employ shrink. Importantly, it so happens that we can scale $K\rightarrow\infty$ (i.e. take an asymptotic in $t$ limit) while maintaining $\nu_K$ constant. The net result is a mean square error given by Eq.~(\ref{eq:constoverhead}) with $c=24.26\pi$, which is a factor of $24.26$ greater than the theoretical optimal value~\cite{gorecki2020pi}, but with the convenient feature that the protocol uses non-adaptive measurements. We refer the interested reader to Ref.~\cite{belliardo2020achieving} for detailed derivation of the results sketched here.

It is also worth noting that other protocols are possible.  For instance, in Ref.~\cite{suzuki2020quantum}, a similar two-step method is described for the estimation of global parameters (i.e.~where the parameter is not restricted to a local neighborhood of parameter space). This protocol provides an explicit method to use some (ultimately negligible) fraction of the sensing time available to narrow down the location of the parameter $q$ in parameter space, followed by an optimal local estimation. We emphasize that the explicit estimation scheme we propose (i.e.~the one in Refs.~\cite{kimmel2015robust,kimmel2015robusterratum,belliardo2020achieving}) does not require adaptive measurements, which is one of its key advantages.

\section{Full Proof of the Main Theorem}\label{sec:part2full}
In this Appendix, we expand on the proof sketch of Theorem 1 in the main text to fully prove the result. For reference, this theorem is restated here. 

\setcounter{theorem}{0}

\begin{theorem}\label{thm:entanglement_supp}
Let $q(\vec\theta)=\vec\alpha\cdot\vec\theta$. Without loss of generality, let $\norm{\vec\alpha}_\infty=|\alpha_1|$. Let $k\in\mathbb{Z}^+$ so that
  \begin{equation}\label{eq:condfull_supp}
      k-1<\frac{\norm{\vec\alpha}_1}{\norm{\vec\alpha}_\infty} \leq k.
 \end{equation}
An optimal protocol to estimate $q(\vec\theta)$, where the parameters $\vec\theta$ are encoded into the probe state via unitary evolution under the Hamiltonian in Eq.~(1) of the main text, requires at least, but no more than, $k$-partite entanglement. 
\end{theorem}
\begin{proof}
We divide our proof into two parts. First, using $k$-partite entangled states from the set of cat-like states considered in the main text, we show the existence of an optimal protocol, subject to the upper bound of Eq.~(\ref{eq:condfull_supp}). Second, we show that there exists no optimal protocol using at most $(k-1)$-partite entanglement, proving the lower bound of Eq.~(\ref{eq:condfull_supp}). 

\noindent\emph{Part 1.} Define $T^{(k)}$ to be the submatrix of $T$ with all columns $n$ such that $\sum_m |T_{mn}|>k$ are eliminated, which enforces that any protocol derived from $T^{(k)}$ uses only states that are at most $k$-partite entangled. Define System $A(k)$ as
 \begin{align}
     T^{(k)}\vec p^{(k)} &=\vec\alpha/\alpha_1, \label{eq:thm1a_supp}\\
     \vec{p}^{(k)}&\geq0 \label{eq:thm1b_supp}.
 \end{align}
 
Let $\vec\alpha'= \vec\alpha/\alpha_1$ and define System $B(k)$ as
\begin{align}
    (T^{(k)})^\top\vec y\geq 0 \label{eq:cond1_supp},\\
    \langle\vec\alpha', \vec y\rangle < 0 \label{eq:cond2_supp}.
\end{align}
By the Farkas-Minkowski lemma~\cite{farkas1902, dinh2014farkas}, System $A(k)$ has a solution if and only if System $B(k)$ does not. In particular, this lemma, which, geometrically, is an application of the hyperplane separation theorem \cite{boyd2004convex} is as follows:
\begin{lemma}[Farkas-Minkowski]\label{lemma:farkas}
Consider the system
\begin{align}
    A \vec{x} &= \vec{b},\label{eqn:farkas_system_1} \\
    \vec{x} &\geq 0, \label{eqn:farkas_system_2}
\end{align}
with $A \in \mathbb{R}^{m\times n}$, $\vec{x} \in \mathbb{R}^{n}$, and $\vec{b}\in\mathbb{R}^{m}$. The above system has a solution if and only if there is no solution $\vec{y}$ to
\begin{align}
    A^{\top}\vec{y} &\geq 0, \\
    \langle \vec{b}, \vec{y} \rangle &< 0.
\end{align}
\end{lemma}

Therefore, to prove the result it is sufficient to show that System $B(k)$ does not have a solution if $\sum_{j>1}|\alpha_{j}'| \leq k-1$, where we used that $\alpha'_1=1$.
We assume that a solution $\vec{y}$ exists and will arrive at a contradiction.  Without loss of generality, we assume that $|y_j| \geq |y_{j+1}|$ for all $1 < j < d$. Eq.~(\ref{eq:cond2_supp}) implies $ \sum_{j>1} \alpha_j' y_j < -y_1$. $(T^{(k)})^\top$ has a row $n^*$ given by $\vec\tau^{(n^*)}=(1, 0, \dots, 0)$, so by Eq.~(\ref{eq:cond1_supp}) any solution $\vec{y}$ to System $B$ has $y_1\geq 0$. Therefore, $\left| \sum_{j>1} \alpha_j' y_j\right| > y_1$, which, by the triangle inequality, implies
\begin{equation}\label{eq:cond2update_supp}
    \sum_{j>1}|\alpha_j'||y_j|>y_1.
\end{equation}
Because $|\alpha_{j}'| \leq 1$ for all $j$, because $\sum_{j>1}|\alpha_{j}'| \leq k-1$, and because $|y_j|$ for $j > 1$ are ordered in descending order, the largest the left-hand-side of Eq.~(\ref{eq:cond2update_supp}) can be is $\sum_{j = 2}^{k} |y_j|$, leading to 
\begin{equation}
\sum_{j = 2}^{k} |y_j|>y_1.
\end{equation}
This directly contradicts Eq.~(\ref{eq:cond1_supp}) for the row of $T^{(k)}$ given by $\vec{\tau} = (1,-\sgn(y_2), \dots, - \sgn(y_{k}), 0, 0, \dots)$.

\noindent\emph{Part 2.}
Using Eq.~(\ref{eq:f1j_final}), we have that, for any optimal protocol,
\begin{align}
    \mathcal{F}(\vec\theta)_{1j}=t \int_0^t ds' \bra{\psi(s')}\hat \sigma^z_1\hat \sigma^z_j\ket{\psi(s')},
\end{align}
where we recall that $\ket{\psi(s)} = U(s)\ket{\psi(0)}$. Because $\bra{\psi(s')}\hat\sigma_{1}^{z}\ket{\psi(s')}=0$ for all $s'$ (see Eq.~(\ref{eq:<g_1>=0})), the integrand is non-zero if and only if $\ket{\psi(s')}$ is such that the first qubit is entangled with the $j$th. Define the indicator variable
\begin{equation}\label{eq:indicators}
    E_{j}(s') = \begin{cases}
    1 & \text{$\ket{\psi(s)}$ entangles qubit $j$ and 1} \\
    0 & \text{else}
    \end{cases},
\end{equation}
for all $j$, including any possible ancilla qubits. Here, we define $E_{1} = 1$ even though the first qubit is not ``entangled'' with itself. Further define
\begin{equation}\label{eq:totalentanglement}
    E(s') = \sum_{j} E_{j}(s') \leq (k-1),
\end{equation}
where $E(s')$ is the total number of sensor qubits entangled with the first qubit at time $s'$ and the upper bound comes from our assumption on the partiteness of our probe states.
We then have that
\begin{align}
    \mathcal{F}(\vec\theta)_{1j}&\leq t \int_0^t ds' E_j(s').
\end{align}

Furthermore, for any optimal protocol using at most $(k-1)$-partite entanglement, we require that
\begin{equation}
    \sum_{j}\left|\frac{\alpha_{j}}{\alpha_{1}}t^{2} \right| = \sum_{j}|\mathcal{F}(\vec{\theta})_{j1}| \leq t\sum_{j}\int_{0}^{t}ds' E_{j}(s') = t\int_{0}^{t}ds'\sum_{j}E_{j}(s) \leq t\int_{0}^{t}ds' (k-1) = (k-1)t^{2}.
\end{equation}
We now have a contradiction, however, as the theorem statement assumed that
\begin{equation}
     \sum_{j}\left|\frac{\alpha_{j}}{\alpha_{1}}t^{2}\right| = \frac{\norm{\vec{\alpha}}_{1}}{\norm{\vec{\alpha}}_{\infty}}t^{2} > (k-1)t^{2}.
\end{equation}
This concludes the proof that $(k-1)$-partite entanglement in any form (i.e., not just from cat-like probe states) is insufficient to generate an optimal protocol.
\end{proof}

We also observe that the lower bound on the size of the least entangled state used in an optimal protocol is really, at its core, a lower bound on the \emph{average} entanglement required to saturate the conditions on the quantum Fisher information matrix. Here, average entanglement refers to weighting the size of the entangled state by the proportion of time it is used in the protocol. This lower bound is simply $\norm{\vec{\alpha}}_{1}/\vec{\alpha}_{\infty}$. The lower bound on the size of the most-entangled state, or the bound on \emph{instantaneous} entanglement, comes from ensuring that this lower bound on average entanglement is achievable (that is, if the instantaneous entanglement is too small at each stage, then the average entanglement required cannot be reached). 

\section{Minimum Entanglement Non-Echoed Protocols}\label{sec:non-echoed}
In this Appendix, we prove that there exist protocols that minimize both instantaneous and average entanglement. We recall from Section~\ref{sec:avg_entanglement} the definition of the non-echoed protocols that minimize average entanglement. 

\begin{definition}[Non-Echoed Protocols]
Consider some $\vec{\alpha}\in\mathbb{R}^d$ encoding a linear function of interest. Let $T$ be the matrix which describes our families of cat-like probe states, and let $\vec{p}$ specify a valid protocol such that $\vec p>0$ and $T\vec{p} = \vec{\alpha}/\norm{\vec{\alpha}}_{\infty}$. We say that the protocol defined by $\vec{p}$ is ``non-echoed'' if $\forall i$ such that $p_{i}$ is strictly greater than 0, $\sgn(T_{ij}) \in \{0, \sgn(\alpha_{j})\}$. 
\end{definition}

We now prove Theorem~\ref{thm:avg_entanglement} from the main text, which we again repeat for simplicity. 

\begin{theorem}\label{thm:avg_entanglement_app}
    For any function encoding $\vec{\alpha}$, there exists a non-echoed optimal protocol with minimum instantaneous entanglement. 
\end{theorem}
\begin{proof}
    We proceed with a relatively simple tweak of the proof of the main theorem. As in that theorem, we assume without loss of generality that $\alpha_{1} = \norm{\vec{\alpha}}_{\infty} = 1$. Also assume, for computational simplicity, that $\alpha_{i>1} < 1$ (i.e. there is only a single maximal-magnitude element of $\vec\alpha$) and that $\alpha_{i} > 0 \, \forall i$. These latter assumptions can easily be lifted, as we describe at the end of the proof. 
    
    We will again use the Farkas-Minkowski lemma~\cite{farkas1902, dinh2014farkas} to show that no vector $\vec{y}$ exists such that
    \begin{align}
        (T^{(k)}_{+})^{\top}\vec{y} &\geq 0, \\
        \langle \vec{\alpha}, \vec{y} \rangle &< 0,
    \end{align}
    proving the existence of a non-echoed protocol. Here, $T^{(k)}_{+}$ is $T$ restricted to non-echoed vectors (i.e., $(T^{(k)}_{+})_{ij} \in \{0,1\}$) with weight at most $k$, where $k = \ceil{\norm{\vec{\alpha}}_{1}}$. Assume a solution $\vec{y}$ exists. Noting that $(T^{(k)}_{+})^\top$ has a row given by $(1,0,\dots,0)$, it must be that $y_{1} \geq 0$. Further, for $\vec{y}$ to be a valid solution, we must have
    \begin{equation}\label{eq:non_echoed_k-1_contradiction}
        \langle \vec{\alpha}, \vec{y} \rangle = \alpha_{1}y_{1} + \sum_{i|i\neq 1, y_{i} \geq 0}\alpha_{i}y_{i} + \sum_{i|y_{i} < 0}\alpha_{i}y_{i} = y_{1} + \sum_{i|i\neq 1, y_{i} \geq 0}\alpha_{i}y_{i} + \sum_{i|y_{i} < 0}\alpha_{i}y_{i} \leq 0.
    \end{equation}
    We proceed with two cases. Suppose that at most $k-1$ elements of $\vec{y}$ are negative. Consider the row of $(T_{+}^{(k)})^{\top}$ that has a 1 in the first index and exactly on the indices where $y_{i} < 0$ (which exists because we have sufficiently restricted the number of negative elements of $\vec{y}$). Then $(T^{(k)}_{+})^{\top}\vec{y} \geq 0$ implies that
    \begin{equation}
        y_{1} + \sum_{i|y_{i}\leq 0}y_{i} \geq 0.
    \end{equation}
    But because $\alpha_{i} < 1$, this immediately implies that
    \begin{equation}
        y_{1} + \sum_{i|y_{i}\leq 0}\alpha_{i}y_{i} \geq 0,
    \end{equation}
    which means that Eq.~(\ref{eq:non_echoed_k-1_contradiction}) cannot be true, yielding a contradiction.

    Now suppose that there are at least $k$ elements of $\vec{y}$ that are negative. Let $S$ be the set of indices corresponding to the $k-1$ largest, in magnitude, $y_{i}$. Then the row of $(T_{+}^{(k)})^{\top}$ with a $1$ in the first index and precisely on the indices in $S$ leads to the condition that
    \begin{equation}
         y_{1} + \sum_{i\in S}y_{i} \geq 0.
    \end{equation}
    However, given the constraint that $\alpha_{i>1} < 1$, we find that
    \begin{equation}
       y_{1} + \sum_{i|i\neq 1, y_{i} \geq 0}\alpha_{i}y_{i} + \sum_{i|y_{i} < 0}\alpha_{i}y_{i} \geq y_{1} + \sum_{i\in S}y_{i} \geq 0,
    \end{equation}
    which is again a contradiction. 

    We briefly comment on how to lift the two assumptions we mentioned earlier. First, in the case where there exist multiple maximal elements, the same argument that generalizes the main theorem will also generalize this argument---see Appendix~\ref{app:generalization}. Second, if we allow $\alpha_{i} < 0$, it is simple to see that a protocol still exists; simply replace $(T_{+}^{(k)})_{ij} = 1$ with $\sgn(\alpha_{i})$ (and leave 0s untouched). 
\end{proof}
Thus, Lemma~\ref{lemma:avg_entanglement} and Theorem~\ref{thm:avg_entanglement} prove there exist protocols that can minimize both instantaneous entanglement (i.e., the maximum size of a cat-like state used in the protocol) and the average entanglement over the course of the entire protocol.

\section{Relaxing the Assumption on a Single Maximum Element}\label{app:generalization}
In this Appendix, we will generalize beyond the assumption in the main text that $|\alpha_1|>|\alpha_j|$ for all $j>1$. Conceptually, nothing is changed by relaxing the assumption, but the algebra becomes somewhat more tedious. In the process, we rigorously derive Eq.~(\ref{eq:bound}) and Eq.~(\ref{eq:saturabilitycond}) of the main text.

\subsection{Generalizing Eq.~(\ref{eq:saturabilitycond}) of the main text}
We start with specifically generalizing Eq.~(\ref{eq:saturabilitycond}). To begin, define
\begin{equation}
L:=\{i\,|\, |\alpha_i|=|\alpha_1|\}.
\end{equation} 
The assumption $|\alpha_1|>|\alpha_j|$ for all $j>1$, stated in the main text, is equivalent to assuming $|L|=1$. For arbitrary size $L$, we have the following set of conditions for the single-parameter bound on $q(\vec\theta)$ to be saturable (Eqs.~(\ref{eq:condonfim}) and (\ref{eq:condonfimoffdiagonal}) of the main text):
\begin{align}\label{eq:condonfim_supp}
    \mathcal{F}(\vec q)_{11}&=\frac{t^2}{\alpha_1^2},\\\label{eq:condonfim_supp_2}
    \mathcal{F}(\vec q)_{1i}&=\mathcal{F}(\vec q)_{i1}=0\quad (\forall\, i\neq 1).
\end{align}
Recall that $\mathcal{F}(\vec q)=J^\top\mathcal{F}(\vec\theta)J$, where $J$ is the Jacobian for the basis transformation from $\vec\theta$ to $\vec q$, $q_1=q$ is the linear function we wish to measure, and the other $q_j$ are some other degrees of freedom we fix. We will show that Eqs.~(\ref{eq:condonfim_supp})-(\ref{eq:condonfim_supp_2}) are satisfied if and only if 
\begin{align}\label{eq:saturabilitycond_supp}
    \sum_{i\in L} \frac{\sgn(\alpha_1)}{\sgn(\alpha_i)}\mathcal{F}(\vec\theta)_{ji}\lambda_i = \frac{\alpha_j}{\alpha_1}t^2,
\end{align}
where $\lambda_i\geq 0$ such that $\sum_i \lambda_i=1$. If $|L|=1$, this reduces to Eq.~(\ref{eq:saturabilitycond}) of the main text. 

It will be important to briefly recount how we obtain the single-parameter bound we are trying to saturate~\cite{eldredge2018optimal,qian2020optimal}. In particular, referring to Eq.~(\ref{eq:boixo}) of the main text, we seek a choice of basis that minimizes $\norm{\hat{g}_q}_s^2$, which will yield the tightest possible bound on $\mathcal{M}$, the mean-square error of $q$. Let us formally define our basis for $\mathbb{R}^{d}$ as $\{\vec\alpha^{(1)}, \vec\alpha^{(2)}, \cdots, \vec\alpha^{(d)}\}$, where $\vec\alpha^{(1)}=\vec\alpha$. We then have that $J^{-1}$ has rows given by these vectors. Let $\{\vec\beta^{(1)}, \vec\beta^{(2)}, \cdots, \vec\beta^{(d)}\}$ be the basis dual to this one. That is, these vectors form the columns of $J$ and satisfy $\vec\alpha^{(i)}\cdot\vec\beta^{(j)}=\delta_{ij}$. We can then write 
\begin{align}
    \vec\theta^\top &=(JJ^{-1}\vec\theta)^\top =(J^{-1}\vec\theta)^\top J^\top,
\end{align}
which allows us to rewrite our Hamiltonian in the convenient form 
\begin{align}
    \hat{H}=\frac{1}{2}\vec\theta^\top\hat{\vec\sigma} + \hat{H}_{c}(s) =\frac{1}{2}\sum_{i=1}^d (\vec\alpha^{(i)}\cdot\vec\theta)\vec\beta^{(i)}\cdot\hat{\vec\sigma} + \hat{H}_{c}(s) ,
\end{align}
where $\hat{\vec\sigma}=(\hat\sigma_1^{z}, \cdots, \hat\sigma_d^{z})^\top$. Then
\begin{align}\label{eq:generator}
    \hat{g}_q(0)=\frac{\partial \hat H}{\partial q}=\frac{\partial \hat H}{\partial (\vec{\alpha}^{(1)}\cdot\vec\theta)}=\frac{\vec\beta\cdot\hat{\vec\sigma}}{2},
\end{align}
where $\vec\beta=\vec\beta^{(1)}$. Because the seminorm is time-independent (see Ref.~\cite{boixo2007generalized}), we immediately have that
\begin{equation}
    \norm{ \hat{g}_q}_s=\norm{ \vec\beta}_1,
\end{equation}
and our tightest bound is given by
\begin{align} \label{eq:lpbeta1}
    &\min_{\vec\beta} \norm{\vec\beta}_1, \nonumber \\
    &\text{s.t.}\, \vec\alpha\cdot\vec\beta=1.
\end{align}
Note that 
\begin{align}\label{eq:min_beta}
    1=\sum_i \alpha_i\beta_i\leq \sum_i |\alpha_i||\beta_i| \leq |\alpha_1| \sum_i |\beta_i| =  |\alpha_1|\norm{\vec\beta}_1.
\end{align}
The first inequality is tight if either $\sgn(\beta_i)=\sgn(\alpha_i)$ or $\beta_{i} = 0$ for all $i$. The second is slightly more complicated to saturate. Recall $L=\{i\,|\, |\alpha_i|=|\alpha_1|\}$. Then the second inequality is tight if and only if 
\begin{align}
&\beta_i=0 \text{ for } i\notin L, \label{eq:0cond}\\
&\sum_{i\in L} |\beta_i|=\frac{1}{|\alpha_1|} \label{eq:n0cond}. 
\end{align}
Any solution $\vec\beta$ specifies the first column of the Jacobian $J$ and allows us to rewrite the conditions in Eq.~(\ref{eq:condonfim_supp})-(\ref{eq:condonfim_supp_2}) as
\begin{align}
     \mathcal{F}(\vec q)_{11}&=\vec\beta^\top \mathcal{F}(\vec\theta)\vec\beta=\frac{t^2}{\alpha_1^2}, \label{eq:cond1new} \\
    \mathcal{F}(\vec q)_{1i}&=\mathcal{F}(\vec q)_{i1}=(\vec\beta^{(i)})^\top\mathcal{F}(\vec\theta)\vec\beta=0 \quad (\forall\, i\neq 1) \label{eq:cond2new}.
\end{align}
As $\vec{\alpha}^{(i)}\cdot\vec{\beta}^{(j)}=\delta_{ij}$, Eq.~(\ref{eq:cond2new}) immediately implies that the vector $\mathcal{F}(\vec\theta)\vec\beta$ must be proportional to $\vec\alpha$ and Eq.~(\ref{eq:cond1new}) specifies the constant of proportionality. In particular, we require
\begin{equation}\label{eq:prop}
    \mathcal{F}(\vec\theta)\vec\beta=\frac{t^2}{\alpha_1^2}\vec\alpha.
\end{equation}
Invoking Eqs.~(\ref{eq:0cond})-(\ref{eq:n0cond}) and the condition that $\sgn(\beta_i)=\sgn(\alpha_i)$ for $\beta_{i} \neq 0$, we write $\beta_i = \lambda_i\sgn(\alpha_i)/|\alpha_1| $, where $\lambda_i \geq 0$ for $i \in L$ and $\lambda_i = 0$ for $i \notin L$ such that $\sum_i\lambda_i = 1$.
The individual components of Eq.~(\ref{eq:prop}) imply
\begin{equation}
    \sum_{i\in L} \mathcal{F}(\vec\theta)_{ij}\sgn(\alpha_i)\lambda_i =   \sum_{i\in L} \mathcal{F}(\vec\theta)_{ji}\sgn(\alpha_i)\lambda_i = \frac{t^2}{|\alpha_1|}\alpha_j,\quad \sum_i \lambda_i=1, \quad \lambda_i\geq 0,
\end{equation}
which, using $|\alpha_1|=\mathrm{sgn}(\alpha_1)\alpha_1$ and that $\mathrm{sgn}(\alpha_1)\mathrm{sgn}(\alpha_i)=\mathrm{sgn}(\alpha_1)/\mathrm{sgn}(\alpha_i)$ for $i \in L$, yields
\begin{equation}\label{eq:7gen}
    \sum_{i\in L} \frac{\sgn(\alpha_1)}{\sgn(\alpha_i)}\mathcal{F}(\vec\theta)_{ij}\lambda_i =   \sum_{i\in L} \frac{\sgn(\alpha_1)}{\sgn(\alpha_i)}\mathcal{F}(\vec\theta)_{ji}\lambda_i = \frac{\alpha_j}{\alpha_1}t^2,\quad \sum_i \lambda_i=1, \quad \lambda_i\geq 0,
\end{equation}
which reduces to Eq.~(\ref{eq:saturabilitycond}) of the main text, when $|L|=1$, as desired.

\subsection{Generalizing the derivation of Eq.~(\ref{eq:prob}) of the main text}
At this point, we can generalize the derivation of Eq.~(\ref{eq:prob}) of the main text to this setting of more than one maximum element of $\vec\alpha$. In particular, Lemma~\ref{lem:f11} can be immediately extended to the following:

\begin{lemma}\label{lem:f11_new}
Any optimal protocol, independent of the choice of control, requires that $\langle \hat{\mathcal{H}}_{j}(t)\rangle = 0$ for all $j\in L$ and that the probe state be of the form
\begin{equation}\label{eq:state_gen}
    \ket{\psi} = \frac{\left(\bigotimes_{j\in L}\ket{b_j}\right)\ket{\varphi_{0}} + e^{i\phi}\left(\bigotimes_{j\in L}\ket{b_j+1}\right)\ket{\varphi_{1}}}{\sqrt{2}},
\end{equation}
for all times $s \in [0, t]$, where
\begin{equation}
    b_j = \begin{cases} 0, \quad \mathrm{if}\,\mathrm{sgn}(\alpha_j)=1,\\
    1, \quad \mathrm{if}\,\mathrm{sgn}(\alpha_j)=-1,
    \end{cases}
\end{equation}
and $\phi, \ket{\varphi_{0}}, \ket{\varphi_{1}}$ can be arbitrary and $s$-dependent. The addition inside the second ket of Eq.~(\ref{eq:state_gen}) is mod 2. 
\end{lemma}
\begin{proof}
We have the following two facts: (1) $\sum_{i\in L} \lambda_i (\mathrm{sgn}(\alpha_j)/\mathrm{sgn}(\alpha_i)) \mathcal{F}(\vec\theta)_{ij}= t^2$ for all $j\in L$ (by Eq.~(\ref{eq:7gen})); (2) $|\mathcal{F}(\vec\theta)_{ij}|\leq \mathcal{F}(\vec\theta)_{jj}$ for all $i$ (by the fact that the Fisher information matrix is positive semidefinite). These facts imply that an optimal protocol must have $\mathcal{F}(\vec\theta)_{jj}=t^2$ for all $j\in L$. The fact that $\langle \hat{\mathcal{H}}_{j}(t)\rangle = 0$ for all $j\in L$ and the fact that all sensors in $L$ must be in a cat-like state over computational basis states follows immediately via an identical calculation to the proof of  Lemma~\ref{lem:f11} for each $j\in L$. From Eq.~(\ref{eq:f1j_final}) it follows directly that these cat-like states over the qubit sensors in $L$ must take the form in the theorem statement in order to achieve the correct sign on the components of $\mathcal{F}(\vec{\theta})$.
\end{proof}

Using Lemma~\ref{lem:f11_new}, it is clear that we should restrict the set $\mathcal{T}$ of states such that $\tau_j^{(n)}=\sgn(\alpha_j)/\sgn(\alpha_1)$ for all $j\in L$ and all $\vec\tau^{(n)}$. This is the generalization of the fact that that, when $|L|=1$, we require $\tau_1^{(n)}=1$ for all $\vec\tau^{(n)}$.

In addition, given the required form of the optimal states, it is easy to generalize Eq.~(\ref{eq:s22}) to the condition that
\begin{equation}\label{eq:s22_new}
    \sum_{i\in L} \left[\lambda_i \int_0^t ds' \bra{\psi(s')}\hat \sigma^z_i\hat \sigma^z_j\ket{\psi(s')}\right] = \frac{\alpha_{j}}{\alpha_{1}}t,
\end{equation}
which implies that, for protocols switching between states in the modified $\mathcal{T}$,
\begin{equation}
    \sum_{i\in L} \left[\lambda_i\sum_{l=0}^{n}(t_{l+1}^{*}-t_{l}^{*})\tau_{j}^{(l)}\right] = \frac{\alpha_{j}}{\alpha_{1}}t,
\end{equation}
where we assume that we switch to the state labeled by $\vec\tau^{(l)}$ at time $t^*_l$. As before, in our protocols $t^*_{l+1}-t^*_l=p_lt$. In addition, $\sum_i \lambda_i=1$. So an optimal protocol requires
\begin{equation}
    t \sum_{l=0}^{n}p_l\tau_{j}^{(l)} = \frac{\alpha_{j}}{\alpha_{1}}t 
    \qquad
    \implies
    \qquad
    T\vec p = \vec\alpha,
\end{equation}
recovering Eq.~(\ref{eq:prob}) of the main text for general $L$, with the addition that we fix $T_{jn}=\tau_j^{(n)}=\sgn(\alpha_j)/\sgn(\alpha_1)$ for all $j\in L$ and all $n$.

\subsection{Generalizing the proof of Theorem 1 of the main text}
Recall, we divided the proof into two parts. First, we showed the existence of an optimal protocol using $k$-partite entangled cat-like states, subject to the upper bound of the theorem statement. Second, we showed that, subject to the lower bound of the theorem statement, there exists no optimal protocol using only $(k-1)$-partite entanglement. 

Let's begin by addressing how the first part changes upon relaxing the assumption that $|\alpha_1|>|\alpha_j|$ for all $j>1$. Note that, given our choice that $\tau_j^{(n)}=\sgn(\alpha_j)/\sgn(\alpha_1)$ for all $j\in L$ and all $\vec\tau^{(n)}$, the first $|L|$ rows of $T^{(k)}$ yield redundant equations in Eq.~(19) of the main text. Therefore, we can define $\tilde T^{(k)}$ as $T^{(k)}$ with all rows $j\in L\setminus \{1\}$ eliminated. Similarly, $\tilde{\vec\alpha}$ is $\vec\alpha$ with elements $j\in L\setminus \{1\}$ eliminated. Further, define the new system of equations, which we call System $\tilde{A}$:
 \begin{align}
     \tilde T^{(k)}\vec \tilde p^{(k)} &=\tilde{\vec\alpha}/\alpha_1, \label{eq:thm1aprime}\\
     \tilde{\vec{p}}^{(k)}&\geq0 \label{eq:thm1bprime}.
 \end{align}
System $A$ has a solution if and only if System $\tilde A$ does.
We can proceed as in the proof in Appendix~\ref{sec:part2full} to show via the Farkas-Minkowski lemma that System $\tilde A$ has a solution if $\norm{\vec\alpha}_1/\norm{\vec\alpha}_\infty\leq k\implies \norm{\tilde{\vec\alpha}}_1/\norm{\tilde{\vec\alpha}}_\infty\leq k-|L|+1$. The details of the proof of this part are completely identical with this substitution. 

The second part of the proof can similarly be adjusted straightforwardly. In particular, to satisfy the condition of Eq.~(\ref{eq:7gen}), which is the generalization of Eq.~(\ref{eq:saturabilitycond}) in the main text, for $j\in L$ we require
\begin{equation}\label{eqn:Fijproofgeneralization}
    \frac{\alpha_j}{\alpha_1}t^2=\frac{\sgn(\alpha_j)}{\sgn(\alpha_1)}t^2=
    \sum_{i\in L} \frac{\sgn(\alpha_1)}{\sgn(\alpha_i)}\mathcal{F}(\vec\theta)_{ij}\lambda_i,
\end{equation}
which implies
\begin{equation}
    t^2=\sum_{i\in L} \frac{\sgn(\alpha_i)}{\sgn(\alpha_j)} \mathcal{F}(\vec\theta)_{ij}\lambda_i.
\end{equation}
This in turn implies that for $i,j\in L$
\begin{equation}
    \mathcal{F}(\vec\theta)_{ij}=\frac{\sgn(\alpha_i)}{\sgn(\alpha_j)}t^2.
\end{equation}
Therefore, for all $i\in L$ we require $\mathcal{F}(\vec\theta)_{ii}=t^2$. From here, arguments identical to those in Appendix~\ref{sec:part2full} apply to all $i\in L$, not just $i=1$. That is, all the probe states must always be fully entangled on the qubits in $L$ and matrix elements $\mathcal{F}(\vec\theta)_{ij}$ for $i\in L$, $j\notin L$ can only accumulate magnitude if sensor $j$ is also entangled with the qubits in $L$. Assuming the existence of an optimal protocol using $(k-1)$-partite entanglement, a contradiction arises in an identical way.

\end{appendix}

\end{document}